\documentclass[11pt]{article}
\usepackage{graphicx} % Required for inserting images
\usepackage{amsmath}
\input{Definitions}

\allowdisplaybreaks

\title{Mechanism Design with Delegated Bidding}

\author{Gagan Aggarwal\\
Google Research\\
\texttt{gagana@google.com}
\and
Marios Mertzanidis \\
Purdue University\\
\texttt{mmertzan@purdue.edu}
\and
Alexandros Psomas\\
Purdue University\\
\texttt{apsomas@purdue.edu}
\and
Di Wang\\
Google Research\\
\texttt{wadi@google.com}
}

\date{}

\begin{document}

\maketitle

\begin{abstract}
    We consider the problem of a designer who wants to allocate resources to representatives, that then distribute the resources they receive among the individuals they represent. Motivated by the way Feeding America, one of the largest U.S. charities, allocates donations to food banks, which then further distribute the donations to food-insecure individuals, we focus on mechanisms that use artificial currencies. We compare auctions through the lens of the Price of Anarchy, with respect to three canonical welfare objectives: utilitarian social welfare (sum of individuals' utilities), Nash social welfare (product of individuals' utilities), and egalitarian social welfare (minimum of individuals' utilities). We prove strong lower bounds on the Price of Anarchy of all auctions that allocate each item to the highest bidder (regardless of how the artificial currency budgets are distributed), subject to a mild technical constraint; this includes Feeding America's current auction, the First-Price auction. In sharp contrast, our main result shows that adapting the classic Trading Post mechanism of Shapley and Shubik to this setting, and coupled with Feeding America's choice of budget rule (each representative gets an amount of artificial currency equal to the number of individuals it represents), achieves a small Price of Anarchy for all generalized $p$-mean objectives (which include the three canonical objectives as special cases) simultaneously. Our bound on the Price of Anarchy of the Trading Post mechanism depends on $\ell$: the product of the rank and the ``incoherence'' of the underlying valuation matrix, which together capture a notion of how ``spread out'' the values of a matrix are. This notion has been extremely influential in the matrix completion literature, and, to the best of our knowledge, has never been used in auction theory prior to our work. Perhaps surprisingly, we prove that the dependence on $\ell$ is necessary: the Price of Anarchy of the Trading Post mechanism is $\Omega(\sqrt{\ell})$.
\end{abstract}

\section{Introduction}

We consider mechanism design in scenarios where every participant in the mechanism is a representative, participating on behalf of (potentially) multiple individuals who are affected by the mechanism's outcome. This situation captures a myriad of real-world scenarios, from social choice (e.g., representatives in the European Parliament vote on behalf of their respective countries' citizens) to advertising (e.g., an auto-bidder in an advertising platform bids on behalf of its campaigns). In this paper, motivated by the way Feeding America, the largest (in terms of revenue) U.S. charity~\cite{Forbes}, allocates its food, we focus on the problem of resource allocation using mechanisms with artificial currencies.

Feeding America, which distributed 5.3 billion meals through its various programs in 2023 alone~\cite{FeedingAmerica}, distributes food to partner food banks across the United States. The consumers of the food donations are, of course, not the food banks themselves, but the individuals the food banks serve. In 2005, Feeding America transitioned from a centralized distribution of food donations to a market-based approach~\cite{prendergast2022allocation}. In the current system, known as the Choice System, each food bank is given an amount of artificial currency (called ``shares'') equal to the population it serves, which is used to submit bids for available (truckloads of) food donations; importantly, a bid here is a single number, despite the fact that the underlying population served might have heterogeneous preferences for the donation.
Given these bids, Feeding America then runs an auction --- the current choice is a First-Price auction --- and the food bank with the highest bid on a donation receives it.\footnote{For small food banks some additional measures are used to address various issues, including joint bidding and short-term credit.} Once a food bank receives food donations via Feeding America's auction, it then allocates these donations among the food-insecure individuals it serves. 

Simply put, the goal of this paper is to inform the choice of auction for the designer (in our motivating example, Feeding America).
We introduce a simple model for studying this question, and compare auctions through the lens of the Price of Anarchy, with respect to the three canonical welfare objectives: utilitarian social welfare (sum of individuals' utilities), Nash social welfare (product of individuals' utilities), and egalitarian social welfare (minimum of individuals' utilities).

\subsection{Our contribution}

We consider the problem of a designer that wants to allocate $m$ divisible items to $C$ strategic centers, where each center $c \in [C]$ represents a set of individuals $I_c$ with additive preferences. The items are allocated to the centers via a mechanism that uses an artificial currency. Such mechanisms are prevalent in theory and practice~\cite{prendergast2022allocation,gorokh2021monetary,budish2017course, walsh2014allocation} because of their simplicity.
Each center receives a \emph{budget} of artificial currency, that depends on the number of centers and their size (i.e., on $|I_1|, \dots, |I_C|$) which it uses to bid on items.
When center $c$ receives a (fractional) allocation of the items, it proceeds to further distribute them among its set of individuals $I_c$ in a manner that maximizes a canonical welfare objective (with respect to individuals in $I_c$); in this paper we focus on the utilitarian social welfare (sum of individuals' utilities), Nash social welfare (product of individuals' utilities), and egalitarian social welfare (minimum of individuals' utilities). We assume that the designer has the same objective as the centers (utilitarian, Nash, or egalitarian social welfare), and is interested in maximizing this objective over all individuals in the system, i.e., $\cup_{c \in C} I_c$. We evaluate different mechanisms by their (pure) Price of Anarchy: the ratio of the optimal objective and the value of the objective in the worst pure Nash equilibrium of the mechanism.

Many natural mechanisms have unnatural equilibria, with a huge Price of Anarchy. For example, a Second-Price auction has an equilibrium where some bidder submits an infinite bid on every item, and all other bidders submit a bid of zero. To rule out such bizarre equilibria we consider mechanisms that satisfy the following constraint: if for a center $c \in [C]$ the sum of bids on the items where $c$ is the highest-bidder (possibly tied for the highest-bid) exceeds the budget allocated to center $c$, then center $c$ does not receive any items. This constraint is satisfied ipso facto by many prominent mechanisms, e.g., the First-Price auction and the All-Pay auction.

We start by showing that a wide family of mechanisms we call ``Highest-Bidder-Wins'' auctions (which captures the First-Price auction, the current auction Feeding America uses in the Choice System) have an unbounded Price of Anarchy for egalitarian and Nash welfare (\Cref{thm: first price lower bound p less than 0}). As the name suggests, these mechanisms allocate each item to the center that bids the highest for it (with an arbitrary tie-breaking rule). Our lower bound holds for all \emph{generalized $p$-mean welfare} objectives for $p \leq 0$, and no matter what rule is used to allocate budgets to centers. 
The generalized $p$-mean of $k$ numbers $x_1, \dots, x_k$ is equal to $\left( \frac{1}{k} \sum_{i \in [k]} x_i^p \right)^{1/p}$.
The $p$-mean family captures a wide range of objectives, including utilitarian social welfare ($p=1$), egalitarian social welfare ($p \rightarrow \infty$), and Nash social welfare ($p=0$).
For the case of utilitarian social welfare, we prove a lower bound of $\Omega(\sqrt{\ell})$ on the Price of Anarchy (\Cref{thm: first price lower bound p is 1}), where $\ell$ is the product of the rank and the incoherence of the underlying valuation matrix (which we explain shortly).

We proceed to study the Price of Anarchy of a different mechanism: the Trading Post mechanism. In the Trading Post mechanism,  originally introduced by Shapley and Shubik~\cite{shapley1977trade} and recently studied in a similar context to ours by Branzei et al.~\cite{branzei2022nash},  
each center $c$ receives a fraction of an item proportional to the amount it bid on it. That is, center $c$ receives a fraction $\frac{b_{j,c}}{\sum_{c' \in [C]} b_{j,c'}}$ of item $j$, where $b_{j,z}$ is the bid of center $z$ on item $j$. We study this mechanism under proportional budgets (each center $c$ receives a budget of $|I_c|$ of an artificial currency), and an all-pay payment rule (thus, the sum of bids is always at most the budget). We prove that the Price of Anarchy of the Trading Post mechanism with proportional budgets, and with respect to all $p$-mean objectives for $p \leq 1$, is at most $\left( 1 + \max_{c \in [C]} \frac{|I_c|}{n} \right) \ell \leq 2 \ell$, where $n = \sum_{c \in [C]} |I_c|$ is the total number of individuals, and $\ell = r \cdot \mu_0$ is the product of the \emph{rank} and the \emph{incoherence} of the underlying valuation matrix (\Cref{theorem: trading post upper bound}); for the case of Nash social welfare ($p=0$) we show an improved upper bound of $2$ (\Cref{theorem:nsw poa}).
Intuitively, small $\ell$ implies that the values of a matrix are ``nicely spread out.'' For example, if $\ell$ is small, we do not have a few values that are significantly bigger than the rest of the values in the matrix. The product of the rank and the incoherence is common in the Randomized Linear Algebra literature, and specifically in the context of Matrix Completion~\cite{candes2012exact,chen2015incoherence,kelner2023matrix}. Rather surprisingly, we prove that small $\ell$ is not only sufficient, but also necessary! We prove a lower bound of $\Omega(\sqrt{\ell})$ on the Price of Anarchy of the Trading Post mechanism with proportional budgets for utilitarian social welfare (\Cref{theorem: welfare poa lower bound}) and egalitarian social welfare (\Cref{theorem: maximin lower bound}), noting that we do not lower bound the Price of Anarchy for Nash social welfare, since the upper bound is already a small constant. Importantly, our lower bounds for $p \leq 0$ on the Price of Anarchy of ``Highest-Bidder-Wins'' auctions hold even for instances where $\ell$ is as small as $2$ (and therefore the comparison with the Trading Post mechanism is accurate).

Overall, a concrete take-home message of our results is that a minor tweak in the allocation rule of Feeding America's auction, while keeping the current choice of budget allocation, drastically improves the Price of Anarchy --- reducing it from infinite to finite --- for all generalized $p$-mean objectives for $p \leq 0$. For utilitarian social welfare ($p=1$), our results show that the performance of the First-Price auction is significantly worse when individuals do not represent themselves: the Price of Anarchy in the presence of direct representation is a small constant, even for subadditive valuations (see Related work), while in our problem it is $\Omega(\sqrt{\ell})$. And, while the Trading Post mechanism shows a similar trend in its Price of Anarchy for this objective,\footnote{Whether a First-Price auction or the Trading Post mechanism achieves a better Price of Anarchy for utilitarian social welfare in our setting is left as an interesting open problem.} our results show that the Trading Post mechanism offers more robust guarantees across different objectives, highlighting its potential as a superior alternative to the First-Price auction in this context.

\vspace{-2mm}
\subsection{Related Work}

\paragraph{The Trading Post Mechanism.} The Trading Post Mechanism, originally introduced by Shapley and Shubik~\cite{shapley1977trade}, has been studied under various names, such as the proportional sharing mechanism~\cite{zhang2005efficiency,feldman2008proportional,christodoulou2015efficiency,gkatzelis2021fair}. 
Johari and Tsitsiklis~\cite{johari2004efficiency} prove that the Price of Anarchy of the Trading Post mechanism, with respect to utilitarian social welfare, is at most $4/3$ for agents with concave valuations. The most relevant work to ours is by Branzei et al.~\cite{branzei2022nash}, who study the Price of Anarchy of the Fisher market mechanism and the Trading Post mechanism, with respect to the Nash social welfare. Branzei et al. consider scenarios where valuations are either perfect substitutes or perfect complements, and show that the Fisher market mechanism has a constant Price of Anarchy for perfect substitutes. However, for perfect complements, this factor degrades linearly with the number of agents. On the other hand, they show that the Trading Post mechanism has a constant Price of Anarchy for any concave valuation function.

\paragraph{Price of Anarchy of Mechanisms with Money.} 
A rich body of literature proves desirable properties of simple mechanisms at their equilibria. 
For example, Daskalakis et al.~\cite{daskalakis2022multi} show that, in multi-parameter settings, the better of simultaneous auctions with entry fees and those with reserved prices are approximately revenue-optimal for additive bidders. Cai et al.~\cite{cai2023simultaneous} extend these findings to subadditive bidder valuations. For the case of welfare maximization, Christodoulou et al.~\cite{christodoulou2016bayesian} prove that simultaneous Second-Price auctions have a Price of Anarchy of 2 for submodular bidders. Similarly, Bhawalkar and Roughgarden~\cite{bhawalkar2011welfare} show that the pure Nash equilibria of simultaneous Second-Price auctions have a Price of Anarchy of at most $2$ for subadditive bidders. Feldman et al.~\cite{feldman2013simultaneous} prove similar results for Bayes-Nash equilibria, for both First-Price and Second-Price auctions.

% \paragraph{Equilibrium Fairness of Mechanisms without Money.} 
% In mechanism design without money, where fairness is typically a first-order concern, strong negative results rule out truthful mechanisms that satisfy non-trivial fairness properties~\cite{amanatidis2017truthful}. To explore the implications of strategic behavior in such problems, Amanatidis et al.~\cite{amanatidis2023allocating} examine the fairness properties of simple mechanisms at their equilibria, and prove that the simple round-robin mechanism produces EF-1 allocations (a standard fairness notion) in equilibrium. Later, Amanatidis et al.~\cite{amanatidis2023round} generalize these results to non-additive valuations.

\paragraph{Mechanisms with Artificial Currencies.} A notable subset of mechanisms without money is mechanisms that use artificial currencies. Such mechanisms are appealing due to their simplicity and have been adopted in many real-world settings (\cite{prendergast2022allocation,budish2017course, walsh2014allocation}). Budish~\cite{budish2011combinatorial} studies scenarios with arbitrary ordinal preferences over items, employing relaxed notions of incentive compatibility to achieve Pareto optimality, approximate market clearing, and envy-freeness. %Jackson and Sonnenschein~\cite{jackson2007overcoming} examine settings where agents participate simultaneously in multiple, one-shot resource allocation mechanisms, ensuring near-optimal welfare at any equilibrium. Their work differs in several key aspects: they assume that for each agent, there exists a prior distribution from which all item valuations are drawn (in an i.i.d. fashion), and they require each agent to report values according to their prior distribution. As the number of items increases to infinity, they demonstrate that the fraction of misreported values approaches zero, and the utility of each agent converges to that of the optimal allocation. 
Gorokh et al.~\cite{gorokh2021monetary} propose a framework that transforms a truthful mechanism with monetary transfers into a truthful dynamic mechanism with no payments, utilizing artificial currencies. In this context, dynamic fair allocation with artificial currencies have also been studied by Gorokh et. al.~\cite{gorokh2021remarkable}, who prove strong performance guarantees of repeated simultaneous First-Price auctions. Finally, Banerjee et al.~\cite{banerjee2023Robust} study such mechanisms in public resource allocation.

\paragraph{Matrix Completion.} 
Matrix Completion is the fundamental problem of reconstructing a matrix given only a few of its inputs. This problem was introduced by Candes and Recht~\cite{candes2012exact} as a candidate solution to the Netflix challenge~\cite{bennett2007kdd}, and has since been extensively studied~\cite{candes2012exact,isinkaye2015recommendation,chen2015incoherence,kelner2023matrix}.

\section{Preliminaries}

We consider the problem of allocating $m$ divisible items among a set $N$ of $|N| = n$ individuals with additive preferences. An allocation $y \in [0,1]^{m n}$ is a partition of items to individuals. An allocation $y$ is feasible if for all $j \in [m]$, $\sum_{i \in N} y_{j,i} \leq 1$, where $y_{j,i}$ is the fraction of item $j$ allocated to individual $i$. Let $\Fcal$ be the set of all feasible allocations.
Each individual $i \in N$, has a value $A_{j,i} \in \R^+$ for item $j \in [m]$.\footnote{We use the notation $A$ for valuations instead of $v$, because (i) we treat the valuations as matrices for the majority of the analysis (therefore, a capital letter is appropriate), and (ii) our analysis uses singular value decompositions (therefore $U$, $\Sigma$, and $V$ are taken).}
Individuals have additive preferences, i.e., the utility of an individual $i \in N$ for an allocation $y$ is $u_{i}(y) = \sum_{j \in [m]} A_{j,i} y_{j,i}$. It will be helpful to think the valuations as an $m \times n$ matrix $A$, where $A_{j,i}$ is the value of individual $i \in [n]$ for item $j \in [m]$, $A^{(i)}$ is the $i^{th}$ column vector, and $A_{(j)}$ is the $j^{th}$ row vector. 
We work with non-scale free objectives (such as welfare), therefore we will assume throughout the paper that individuals' utilities are normalized. Specifically, we assume that the valuation vector of individual $i$ has \emph{Euclidean norm} equal to $1$, i.e., $\norm{A^{(i)}}_2 = 1$, for all $i \in N$.%\footnote{The standard normalization in fair division is that the value for the grand bundle is $1$, i.e., $\norm{A^{(i)}}_1 = 1$.}

In our problem, individuals cannot directly participate in a mechanism. There are $C$ centers, and each individual $i \in N$ is represented by exactly one center $c \in [C]$. Let $I_c$ be the set of individuals that center $c$ represents; $N = \cup_{c \in [C]} I_c$, i.e., every individual has a representative center. We write $A_{j,(c,i)} \in \R^+$ for the value of individual $i \in I_c$ for item $j \in [m]$, and  and $u_{(c,i)}(y)$ for the utility of individual $i \in I_c$ for allocation $y$.
% Each individual $i \in I_c$, has a value $A^c_{j,i} \in \R^+$ for item $j \in [m]$.\footnote{We use the notation $A$ for valuations instead of $v$, because (i) we treat the valuations as matrices for the majority of the analysis (therefore, a capital letter is appropriate), and (ii) our analysis uses singular value decompositions (therefore $U$, $\Sigma$, and $V$ are taken).} Individuals have additive valuations; the utility of an individual $i \in I_c$ for a partition $y$ is $u_{c,i}(y) = \sum_{j \in [m]} A_{j,i}^c y_{c,i}^j$. 
The utility of a center $c \in [C]$ for an allocation $y$ is the Generalized $p$-Means Welfare, 
$GMW_{c,p}(y) = \left( \frac{1}{|I_c|} \sum_{i \in I_c} u_{c,i}(y)^p \right)^{1/p}$, for some $p \le 1$. We assume that all centers have the same Generalized $p$-Means Welfare objective (i.e. agree on the choice of $p$).
A problem instance $\Ical = (\{I_c\}_{c \in [C]}, A, p)$ is therefore characterized by the sets of individuals $\{I_c\}_{c \in [C]}$, their valuations $A \in \R^{m n}$, and $p \leq 1$.

\vspace{2mm}
\noindent \textbf{Mechanisms.} We allocate items via a mechanism $\Mcal$ that centers participate in. A mechanism $\Mcal$ is characterized by an allocation rule $x^{\Mcal}(\cdot): \R^{mC} \rightarrow  [0,1]^{mC}$, a ``budget rule'' $g^{\Mcal}(\cdot): \R^{C}_+ \rightarrow  \R^{C}_+$, and a ``virtual payment rule'' $q^{\Mcal}(\cdot): \R^{mC} \rightarrow  \R^{C}$.
The allocation rule elicits from each center $c \in [C]$ a bid $b_{j,c}$ for item $j$ and outputs an allocation of items to centers; the virtual payment rule elicits bids and outputs a (virtual) payment for each center.  The budget rule $g^{\Mcal}(|I_1|, \dots, |I_{C}|)$ maps centers' population to a vector of budgets, one for each center. If the payment of a center $c$ exceeds its budget $g^{\Mcal}_c$, we assume that the center is not allocated any items. Given an allocation of the items, each center divides its items among the individuals it represents in a way that the Generalized $p$-Means Welfare is maximized. It will be convenient to equivalently think of $x^{\Mcal}$ as eliciting, for each individual $i \in I_c$ and each item $j \in [m]$, a bid $b_{j,(c,i)}$ that center $c$ chooses on the individual's behalf, and outputting an allocation $x(b)$ of items to individuals.
We assume that centers are strategic: a center $c \in [C]$, given the bids of other centers $b_{-c}$,  will report a bid $b_c$ that maximizes its utility (i.e. $b_c = \argmax_{b \in \R^{m|I_c|}} GMW_c(x(b ; b_{-c}))$). 

\vspace{2mm}
\noindent \textbf{Bidding Restrictions.} 
Of course, we'd like to prove Price of Anarchy bounds for every mechanism. However, many natural mechanisms have unnatural equilibria with a huge Price of Anarchy. For example, a Second-Price auction has an equilibrium where some bidder submits an infinite bid on every item, and all other bidders submit a bid of zero. In order to rule out such bizarre equilibria we consider mechanisms that satisfy the following constraint: if, for a center $c \in [C]$, the sum of bids on the items where $c$ is the highest-bidder (possibly tied) exceeds the budget of center $c$, then center $c$ does not receive any items. Formally, if $\sum_{j \in [m]} b_{j,c} \cdot 1\{ b_{j,c} \geq b_{j,c'}, \forall c' \in [C] \} > g^{\Mcal}_c(|I_1|, \dots, |I_C|)$, center $c$ does not receive any items. This constraint is satisfied by many prominent mechanisms, e.g., the First-Price auction and the All-Pay auction.

\vspace{2mm}
\noindent \textbf{Highest-Bidder-Wins Auctions.} In a Highest-Bidder-Wins auction with budget rule $g$, denoted by $\text{Highest-Bidder-Wins}_g$, each center $c \in [C]$ gets a budget of $g_c(|I_1|, \dots, |I_C|)$.
Then, the allocation rule dictates that item $j$ is allocated to the center with the highest bid, i.e., center $c$ gets the entirety of item $j$ if $b_{j,c} > b_{j,c'}$, for all centers $c' \neq c$. Our negative results hold for any tie-breaking rule when there are ties for the highest bid. If no center bids for an item, then the item is not allocated.

\vspace{2mm}
\noindent \textbf{The Trading Post Mechanism.}
In the Trading Post mechanism with proportional budgets, each center $c \in [C]$ gets a budget of $g_c(|I_1|, \dots, |I_C|) = |I_c|$. Then, the amount of item $j$ allocated to center $c$ is $b_{j,c}/\sum_{c' \in [C]} b_{j,c'}$, where $b_{j,c'}$ is the amount that center $c'$ bids on item $j$. It will be convenient to think of $b_{j,(c,i)}$, the amount that center $c$ bids on behalf of agent $i \in I_c$ for item $j$. The virtual payment rule is ``all-pay,'' i.e., center $c$ pays $\sum_{j \in [m]} b_{j,c}$.

\vspace{2mm}

\noindent \textbf{Matrix Coherence.}
Let $A$ be an $r$-rank matrix. Then $A$ can be broken down into its singular value decomposition, $A = U \Sigma V^T$ where $U$ and $V$ are $m \times r$, $|N| \times r$  orthonormal matrices respectively.
The following definition is a measure of how ``spread out'' the values of a matrix are.

\begin{definition}[Coherence]
    Let $U$ be a subspace of $\R^m$ of dimension $r$ and $\mathbf{P}_U$ be the orthogonal projection onto $U$. Then the coherence of $U$ is defined as $\mu(U) = \frac{m}{r} \max_{i \in [m]} \norm{\mathbf{P}_U \cdot  e_i}^2_2$, where $e_i$ is the $i^{th}$ standard basis vector.
\end{definition}

\begin{definition}[$\mu_0$-incoherence]\label{definition:incoherence}
    We say that a valuation matrix $A$ with a singular value decomposition of $U \Sigma V^T$, is $\mu_0$-incoherent if $\mu(U) \le \mu_0$.
\end{definition}

Intuitively, the smaller the product of the rank and the incoherence of a matrix, the more ``spread out'' the information is. For example, if this product is small, we do not have a few values that are significantly bigger than the rest of the values in the matrix. A valuation matrix with small rank and incoherence, is such that there are no individuals that have high values only for a small subset of the items. This``spread out'' property proves valuable in establishing desirable properties of equilibria, as it ensures significant competition across all items, while allowing various allocations to achieve near-optimal results. Conversely, when individuals value only a small subset of items, it becomes possible to construct pathological instances where some individuals are left severely under-served. We note that our definition of incoherence is a relaxed version of the standard definition used in the matrix completion literature, which also requires $\mu(V) \le \mu_0$~\cite{candes2012exact}. Throughout the paper, we use $\ell = \mu_0 \, r$ for the product of the rank and the incoherence of a matrix.

\vspace{2mm}
\noindent \textbf{Price of Anarchy.}
Our goal is to design a mechanism $\Mcal$ that maximizes the Generalized $p$-Means Welfare of individuals, $GMW_p(y) = \left( \frac{1}{|N|} \sum_{c \in [C]} \sum_{i \in I_c} u_{(c,i)}(y)^p \right)^{1/p}$. 
We measure the quality of a mechanism $\Mcal$ via its Price of Anarchy (PoA)~\cite{koutsoupias2009worst}: the ratio of the optimal Generalized $p$-Means Welfare divided by the Generalized $p$-Means Welfare of $\Mcal$, in the worst-case \emph{pure} Nash equilibrium.
We say that $\Ical \in \mathcal{W}_{\ell}$ if the valuation matrix (in matrix form) is $\mu_0$-incoherent and has rank $r$ such that $\ell = \mu_0 \, r$.
A problem instance $\Ical$ and some mechanism $\Mcal$ yield a set of pure Nash equilibria $E$. Let $y^* \in \argmax_{y \in \Fcal} GMW_p(y)$ be the allocation that maximizes our objective. Then the price of anarchy of our mechanism $\Mcal$ is
$
PoA_{p, \ell}(\Mcal) = \max_{\Ical \in \mathcal{W}_{\ell}, \tilde{y} \in E} \left\{ \frac{GMW_p(y^*)}{GMW_p(\tilde{y})}\right\}.
$

% and for each individual $i \in I_c$ and item $j \in [m]$, $x_{j,(c,i)}(b) = \frac{b_{j,(c,i)}}{\sum_{c \in [C]} \sum_{i \in I_c} b_{j,(c,i)}}$ if $\sum_{i \in I_c} \sum_{j \in [m]} b_{j,(c,i)} \le B_c$, otherwise $x^j_{c,i}(b)=0$.

\section{Negative Results for Highest-Bidder-Wins Auctions}

In this section, we prove that Highest-Bidder-Wins auctions have a large Price of Anarchy with respect to the $p$-mean welfare, no matter how the designer chooses to allocate budgets to the centers. For $p \leq 0$, the PoA is unbounded, even when $\ell$, the product of the rank and the incoherence of the underlying valuation matrix, is as small as $2$; for $p=1$, the PoA is at least $\Omega(\sqrt{\ell})$. All missing proofs are deferred to~\Cref{appendix: missing}.

The following two lemmas are crucial in our lower-bound constructions (both for Highest-Bidder-Wins auctions, as well as for the Trading Post mechanism). In the proofs we explicitly calculate the incoherence and rank of certain (complicated) matrices.

\begin{lemma}\label{lemma:welfare example incoherence}
    Let $A$ be an $m \times n$ matrix, and $k \ge 1$ some integer that perfectly divides $m$. Let for all $i \in [m/k]$, and $j \in \{(i-1) \, k+1, \cdots , i \, k\}$, $A_{j,i} = 1/\sqrt{k}$, while for $j\le (i-1) \, k$ and $j> i \, k$, $A_{j,i} = 0$. For $i \in \{m/k+1, \cdots, n\}$, and $j \in [m]$, $A_{j,i} = 1/\sqrt{m}$. Then $A$ has rank $m/k$ and is $1$-incoherent.
%     \begin{equation*}
%     A = \begin{bmatrix}
% \frac{1}{\sqrt{k}} & 0 & \cdots  &\frac{1}{\sqrt{m}} & \cdots & \frac{1}{\sqrt{m}}\\
% \frac{1}{\sqrt{k}} &  0 & \cdots  &\frac{1}{\sqrt{m}} & \cdots & \frac{1}{\sqrt{m}}\\
% \vdots & \vdots  & \ddots & \vdots & \ddots & \vdots \\
% 0 & \frac{1}{\sqrt{k}} & \cdots  &\frac{1}{\sqrt{m}} & \cdots & \frac{1}{\sqrt{m}}\\
% \vdots & \vdots & \ddots & \vdots & \ddots & \vdots \\
% 0 & 0 & \cdots & \frac{1}{\sqrt{m}} & \cdots & \frac{1}{\sqrt{m}}
% \end{bmatrix}. 
% \end{equation*}
\end{lemma}

\begin{lemma}\label{lemma:maxmin example incoherence}
    Let $A$ be an $m \times n$ matrix, and $k \ge 1$ some integer that perfectly divides $m$. Let for all $j \in [k]$ $A_{j,1} = 1/\sqrt{k}$, while for all $k+1 \le j \le m$, $A_{j,1} = 0$. Also, for all $i \in [n] - \{1\}$, and all $j \in [m]$ $A_{j,i} = 1/\sqrt{m}$. Then $A$ has rank $2$, and is $\frac{m}{2k}$-incoherent.
%     \begin{equation*}
%     A = \begin{bmatrix}
% \frac{1}{\sqrt{k}} & \frac{1}{\sqrt{m}} & \cdots & \frac{1}{\sqrt{m}}\\
% \frac{1}{\sqrt{k}} & \frac{1}{\sqrt{m}} & \cdots & \frac{1}{\sqrt{m}}\\
% \vdots & \vdots & \ddots & \vdots \\
% 0 & \frac{1}{\sqrt{m}} & \cdots & \frac{1}{\sqrt{m}}
% \end{bmatrix}.
% \end{equation*}
\end{lemma}

\begin{theorem}\label{thm: first price lower bound p less than 0}
For all  $p \le 0$ and all budget rules $g: \mathbb{R}_+^C \rightarrow \mathbb{R}_+^C$, $PoA_{p,2}(\text{Highest-Bidder-Wins}_g) \rightarrow \infty$.
\end{theorem}

\begin{proof}
    Towards a contradiction, assume that the PoA of a Highest-Bidder-Wins auction coupled with some tie-breaking rule --- let $\Mcal$ denote this auction --- is bounded.

    Consider the following example where $C=m$. Each center $c \in [C]$ represents exactly two individuals. Without loss of generality, since all centers have the same number of individuals, assume that $g^{\Mcal}_1(2, \cdots, 2) \in \argmin_{c \in [C]} g^{\Mcal}_c(2, \cdots, 2)$.
    Individual $1$ of center one like the first $m/2$ items equally, i.e., $v_{(1,1)} = [\sqrt{2/m}, \ldots, \sqrt{2/m}, 0, \ldots , 0]^T$; individual $2$ of center one likes the last $m/2$ items equally, i.e., $v_{(1,2)} = [0, \ldots , 0, \sqrt{2/m}, \ldots, \sqrt{2/m}]^T$. The remaining individuals value all items equally, i.e., $\forall c \in [C]/\{1\}, i \in I_c, v_{(c,i)} = [\frac{1}{\sqrt{m}}, \ldots, \frac{1}{\sqrt{m}}]^T$. The induced valuation matrix is precisely the valuation matrix of~\Cref{lemma:welfare example incoherence} for $k = m/2$; therefore, we have that $r=2$ and $\mu_0 = 1$, and thus $\ell = 2$, for the induced valuation matrix.  %Since individuals in centers $c\ge 2$ value all items equally, we say that center $c$ (for $c\geq 2$) ``receives at least $z$ items'' if the sum of the fractions of items received by center $c$ is at least $z$. 

    Any outcome that results in some individual getting zero utility is arbitrarily bad, for any $p \leq 0$, since $GMW_p(.)$ would be equal to zero, and an optimal with respect to $GMW_p(.)$ solution has strictly positive value. Thus, in order for our instance to have a bounded PoA, center one must receive (partially) at least one item $j \in \{1, \cdots, m/2\}$ and one item $j' \in \{m/2 + 1, \cdots, m\}$. For this to occur in a Highest-Bidder-Wins auction, it must be that center one has the highest bid on item $j$ and on item $j'$, noting that if the allocation of $j$/$j'$ is fractional, then some other center is tied with center one for the highest bid, and the fraction depends on the specific tie-breaking rule. Note that it must be the case that center one has a strictly positive bid on both items, and thus each individual bid must be strictly less than its budget.

    \noindent \textbf{Case 1:} Assume that there exists a center $c' \geq 2$ such that $g^{\Mcal}_{c'}(2, \cdots, 2)> g^{\Mcal}_{1}(2, \cdots, 2)$. Consider an arbitrary equilibrium of $\Mcal$. It must be the case that in this equilibrium center $c'$ receives a total fraction of items of at least two, otherwise, $c'$ could outbid center one on both items $j$ and $j'$ (that has the current highest bid on those items, by the bounded PoA assumption) and receive them both. However, since center $c'$ receives a total fraction of items of at least two and center one receives strictly positive fractions of some items, there must exist a center $c^*$ that receives a total fraction of strictly less than one item. However, $c^*$ could outbid center one on item $j$ or $j'$ (since center one has the smallest budget) and receive it in full, contradicting the equilibrium property. 
    
    \noindent \textbf{Case 2:} Assume that $\forall c \in [C]$, $g^{\Mcal}_{c}(2, \cdots, 2) = g^{\Mcal}_{1}(2, \cdots, 2)$. In this case, consider the situation where every center $c$ bids its entire budget on item $j=c$. It is easy to confirm that this is a Nash equilibrium, and center one receives only one item (and thus the PoA is unbounded).
\end{proof}

\begin{theorem}\label{thm: first price lower bound p is 1}
For  all budget rules $g: \mathbb{R}_+^C \rightarrow \mathbb{R}_+^C$, $PoA_{1,\ell}(\text{Highest-Bidder-Wins}_g) \in \Omega(\sqrt{\ell})$.
\end{theorem}

\section{The Trading Post Mechanism}

In this section, we analyze the Price of Anarchy of the Trading Post mechanism with proportional budgets. In~\Cref{sec: lemmas} we prove some technical lemmas. In~\Cref{section: bounding poa} we prove upper bounds on the Price of Anarchy, and in~\Cref{section:lower bounds} we prove lower bounds.

Missing proofs can be found in~\Cref{appendix: missing}.

%, where each center $c \in [C]$ gets a budget of $|I_c|$, and for each individual $i \in I_c$ and item $j \in [m]$, $x_{j,(c,i)}(b) = \frac{b_{j,(c,i)}}{\sum_{c \in [C]} \sum_{i \in I_c} b_{j,(c,i)}}$ if $\sum_{i \in I_c} \sum_{j \in [m]} b_{j,(c,i)} \le |I_c|$, otherwise $x^j_{c,i}(b) = 0$. %For the sake of our analysis, in this section, we will also assume that the valuations of individuals for items are normalized so that $\norm{v_{c,i}}_2 = 1$.

% \ap{something is a bit strange here. we have $x_{j,(c,i)}(b)$ in the mechanism, which should only make things worse (since centers can get even better utility after receiving an allocation). double-check that everything is ok here... (especially the lower bounds)}

% \mm{The two things are equivalent since the center dictates the bids of each individual. Lets say that the center would bid $b$ for an item and then give 1/3 to one individual and 2/3 to an other. It is equivalent with making one individual directly bid $1/3 \; b$ and the the other to bid $2/3 \; b$}

% \ap{at the end of the day the mechanism only gives stuff to the centers. maybe the center doesn't want to give the items in this way (because there is a better way to optimize for $GWM_p$). for upper bounds the assumption seems fine, but for lower bounds it's not. I checked though, and in the lower bounds we are doing the right thing}

\subsection{Implications of Incoherence}\label{sec: lemmas}

Here, we prove a few technical lemmas that are useful when analyzing the Trading Post mechanism.

\begin{lemma} \label{lemma:incoherence valuation bound}
    Let $A$ be an $m \times n$ matrix with rank $r$ that is $\mu_0$-incoherent, and $\forall i \in [n],  \norm{A^{(i)}}_2 =1$. Then, $\max_{ j\in [m], i \in [n]} A_{j,i} \le \sqrt{\frac{\mu_0 r}{m}}$.
\end{lemma}

\begin{proof}
    Let $U \Sigma V^T$ be the singular value decomposition of $A$. By definition, for all $i \in [n]$, $A^{(i)}$ is a linear combination of the columns of $U$. Let  $A^{(i)} = U \lambda_i$, where $\lambda_i = \Sigma V^T e_i$. Since $\norm{A^{(i)}}_2 = 1$ we have that $\norm{U\lambda_i}_2^2 = 1$. The vectors comprising $U$ are pairwise orthogonal, therefore, by the multi-dimensional generalization of Pythagoras's theorem, we have that $\sum_{\ell \in [r]} \norm{U^{(\ell)} \lambda_{i,\ell}}_2^2 = \norm{U\lambda_i}_2^2 = 1$, where $U^{(\ell)}$ is the $\ell^{th}$ column of $U$. Let $U_{(j)}$ be the $j^{th}$ row of $U$. Then $A_{j,i} = U_{(j)} \lambda_i \le \sqrt{\norm{U_{(j)}}_2^2 \norm{\lambda_i}_2^2} = \norm{U_{(j)}}_2$, where we used the Cauchy-Schwarz inequality.
    Furthermore, note that since $U$ is orthonormal, $\mathbf{P}_U = U U^T$, and thus,
\begin{equation*}
  \mu(U) = \frac{m}{r} \max_{i \in [m]} \norm{UU^T e_i}^2_2 = \frac{m}{r} \max_{i \in [m]} \left(UU^T e_i\right)^T UU^T e_i =  \frac{m}{r} \max_{i \in [m]} e_i^T U  U^T e_i = \frac{m}{r} \max_{i \in [m]} \norm{U^T e_i}^2_2.  
\end{equation*}
    Since $A$ is $\mu_0$-incoherent, we have $A_{j,i} \le \norm{U_{(j)}}_2 \le \max_{j \in [m]} \norm{U^T e_j}_2 \le \sqrt{\frac{\mu_0 r}{m}}$.
\end{proof}

\begin{lemma} \label{lemma: technical bound}
    Let $y \in [0,1]^{k}$, be such that, $\sum_{i \in [k]} y_i \le m$, then for all $p \neq 0$,
    \[\left( \sum_{i \in [k]} \left( y_{i}\right)^{p} \right)^{1/p} \le \frac{m}{k^{1-1/p}}.\]
\end{lemma}

\begin{lemma} \label{lemma: equal split bound}
    Let $A$ be an $m \times n$, $r$-rank, and $\mu_0$-incoherent valuation matrix. Then, for all $p \neq 0$, $\frac{GMW_p(y^*)}{GMW_p(y^{equal})} \le \mu_0 \, r $, where $y^* \in \argmax_{y \in \Fcal} GMW_p(y)$, and $y^{equal}$ is the equal split allocation (i.e. $j \in [m], i \in [n]$, $y^{equal}_{j,i} = \frac{1}{n}$).
\end{lemma}

\begin{proof}
    For $GMW_p(y^*)$ we have
\begin{align*}
    GMW_p(y^*) &= \left(\sum_{c \in [C]} \sum_{i \in I_c} \frac{1}{n}(u_{(c,i)}(y^*))^p\right)^{1/p}\\
    &= \left(\sum_{c \in [C]} \sum_{i \in I_c} \frac{1}{n}\left(\sum_{j \in [m]} A_{j,(c,i)} y^*_{j,(c,i)}\right)^p\right)^{1/p}\\
    &\le \left(\sum_{c \in [C]} \sum_{i \in I_c} \frac{1}{n}\left(\sum_{j \in [m]} \sqrt{\frac{\mu_0 r}{m}} y^*_{j,(c,i)}\right)^p\right)^{1/p} \tag{\Cref{lemma:incoherence valuation bound}} \\
    &=  \sqrt{\frac{\mu_0 r}{m}} \left( \frac{1}{n}\right)^{1/p}\left(\sum_{c \in [C]} \sum_{i \in I_c} \left(\sum_{j \in [m]} y^*_{j,(c,i)}\right)^{p} \right)^{1/p}\\
    &\le  \sqrt{\frac{\mu_0 r}{m}} \left( \frac{1}{n}\right)^{1/p}\frac{m}{n^{1-1/p}} \tag{\Cref{lemma: technical bound}}\\
    &= \sqrt{\frac{\mu_0 r}{m}} \frac{m}{n} \\
    &=\frac{\sqrt{m \mu_0 r}}{n}.
\end{align*}
It remains to lower bound $GMW_p(y^{equal})$, given the incoherence assumption on $A$. First, it is easy to confirm that $u_{(c,i)}(y^{equal}) = \frac{1}{n} \sum_{j \in [m]} A_{j,(c,i)}$. The solution to the following program is a lower bound on $n \, u_{(c,i)}(y^{equal})$:

% \begin{equation*}
% \begin{array}{ll@{}ll}
% \text{minimize}  &\displaystyle\sum_{j \in [m]} v_j\\
% \\
% \vspace{7pt}
% \text{s.t.}& \sum_{j \in [m]} v_j^2 = 1  \\
% \vspace{7pt}

% & v_j \le \sqrt{\frac{\mu_0 r}{m}} &&\forall j \in [m]\\
% \vspace{7pt}

% & v_j \ge 0 &&\forall j \in [m]

% \end{array}
% \end{equation*}

\begin{equation*}
\begin{array}{ll}
\text{minimize} & \displaystyle\sum_{j \in [m]} v_j \quad \text{s.t.} \quad \sum_{j \in [m]} v_j^2 = 1, \quad v_j \le \sqrt{\frac{\mu_0 r}{m}} \; \forall j \in [m], \quad v_j \ge 0 \; \forall j \in [m].
\end{array}
\end{equation*}

This program is minimized when there are $\frac{m}{\mu_0 r}$ variables with values $\sqrt{\frac{\mu_0 r}{m}}$ and the remaining variables are set to $0$. Thus $u_{(c,i)}(y^{equal}) \ge \frac{1}{\sqrt{\mu_0 r}} \frac{\sqrt{m}}{n}$. Therefore, we have that:
\begin{align*}
    GMW_p(y^{equal}) &= \left(\sum_{c \in [C]} \sum_{i \in I_c} \frac{1}{n}(u_{(c,i)}(y^{equal}))^p\right)^{1/p} \ge \left(\sum_{c \in [C]} \sum_{i \in I_c} \frac{1}{n}\left(\frac{1}{\sqrt{\mu_0 r}} \frac{\sqrt{m}}{n}\right)^p\right)^{1/p} = \frac{1}{\sqrt{\mu_0 r}} \frac{\sqrt{m}}{n}.
\end{align*}
Combining the above we have
\[\frac{GMW_{p}(y^*)}{GMW_{p}(y^{equal})} \le \frac{\frac{\sqrt{m \mu_0 r}}{n}}{\frac{1}{\sqrt{\mu_0 r}} \frac{\sqrt{m}}{n}} = \mu_0 r. \qedhere\]
\end{proof}

\subsection{Price of Anarchy of the Trading Post Mechanism: Upper Bounds}
\label{section: bounding poa}

In this section, we prove our main positive results: upper bounds on the price of anarchy of the Trading Post mechanism. Recall that $\ell = \mu_0 \, r$ is the product of the rank and the incoherence of the underlying valuation matrix, $n$ is the number of individuals in the system, and $|I_c|$ is the number of individuals that center $c$ represents. We handle $p\neq 0$ (\Cref{theorem: trading post upper bound}) and $p=0$ (\Cref{theorem:nsw poa}) separately, where we show bounds of $\left( 1 + \max_{c \in [C]} \frac{|I_c|}{n} \right) \, \ell \leq 2 \ell$, and $2$, respectively.

\begin{theorem} \label{theorem: trading post upper bound}
    $PoA_{p, \ell}(\text{Trading Post}) \le \left( 1 + \max_{c \in [C]} \frac{|I_c|}{n} \right) \, \ell$, for all $p \leq 1$, $p \neq 0$.
\end{theorem}

\begin{proof}
Let $\tilde{y}$ be the allocation at an arbitrary equilibrium of the Trading Post mechanism, and $y^* \in \argmax_{y \in \Fcal} GMW_p(y)$ be an allocation that maximizes $GMW_p(.)$ over all feasible allocations. Our main goal will be to prove that
\begin{equation}\label{eq: main bound for trading post}
    \frac{GMW_p(y^{equal})}{GMW_p(\tilde{y})} \le 1 + \max_{c \in [C]} \frac{|I_c|}{n},
\end{equation}
where $y^{equal}$ is the equal split allocation (i.e. $y^{equal}_{j,(c,i)} = \frac{1}{n}$). 
Given~\Cref{eq: main bound for trading post}, and using the fact that $\frac{GMW_p(y^*)}{GMW_p(y^{equal})} \le \mu_0 \, r$ (\Cref{lemma: equal split bound}) we can prove the theorem:
\[
PoA_{p,\ell}(TP) =\frac{GMW_p(y^*)}{GMW_p(\tilde{y})} = \frac{GMW_p(y^*)}{GMW_p(y^{equal})}\frac{GMW_p(y^{equal})}{GMW_p(\tilde{y})} \le (1 + \max_{c \in [C]} \frac{|I_c|}{n}) \, \mu_0 \, r = (1 + \max_{c \in [C]} \frac{|I_c|}{n}) \, \ell.\]

To establish~\Cref{eq: main bound for trading post}, we start by showing that every center $c$ has a bidding strategy that, assuming that all other centers bid their equilibrium bids, guarantees an allocation that is a scalar multiple of the equal split allocation. Since this deviation cannot increase the utility of the center (by the equilibrium property), this implies that the equilibrium utility is at least a scalar multiple of the utility in the equal split allocation. Combining with a lower bound on this scalar we can establish~\Cref{eq: main bound for trading post}, concluding the proof.

Let $\tilde{b}$ be the equilibrium bids (which yield allocation $\tilde{y}$ in the Trading Post mechanism). Let $\tilde{p}_j = \sum_{c \in [C]} \sum_{i \in I_c} \tilde{b}_{j,(c,i)}$ be the sum of bids on item $j \in [m]$, and $(\tilde{p}_{-c})_j = \tilde{p}_j - \sum_{i \in I_c} \tilde{b}_{j,(c,i)}$ be the sum of bids on item $j \in [m]$ excluding the bids of center $c$.

% \ap{i've replace $y$ with $y^{equal}$. we should be able to simplify the math...}

% \ap{I've made the simplifications. Replace if everything looks correct to you.}

Consider the following function: $f_c(x) = \sum_{j \in [m]} \frac{ |I_c| (\tilde{p}_{-c})_j }{n x - |I_c|}.$
In the interval $\left( |I_C|/n , \infty\right)$, $f_c(x)$ is continuous, decreasing and $\lim_{x \rightarrow |I_c|/n^+} f_c(x) = \infty, \lim_{x \rightarrow \infty } f_c(x) = 0$.
Thus, there exists a point $\beta_c  > |I_c|/n$ such that $f_c(\beta_c) = |I_c|$. 

Next, consider the following deviation for center $c$. The bid of individual $i \in I_c$ for item $j \in [m]$ is %$b'_{j,(c,i)} = \frac{y^{equal}_{j,(c,i)} (\tilde{p}_{-c})_j}{\beta_c - \sum_{i' \in I_c} y^{equal}_{j,(c,i')}}$ \ap{ 
$b'_{j,(c,i)} = \frac{(\tilde{p}_{-c})_j}{n \beta_c - |I_c| }$. %}
We prove that this is a valid bidding strategy, i.e., (i) all bids are non-negative, and (ii) the sum of bids is less or equal to the budget $|I_c|$. Condition (i) is implied by the fact that $\beta_c > \max_{j \in [m]}\left\{\sum_{i \in I_c} y^{equal}_{j,(c,i)} \right\}$. Regarding condition (ii) we have that:

% \begin{align*}
%     \sum_{j \in [m]} \sum_{i \in I_c} b'_{j,(c,i)} &=  \sum_{j \in [m]} \sum_{i \in I_c} \frac{y^{equal}_{j,(c,i)} (\tilde{p}_{-c})_j}{\beta_c - \sum_{i' \in I_c} y^{equal}_{j,(c,i')}} = \sum_{j \in [m]}  \frac{ (\tilde{p}_{-c})_j \sum_{i \in I_c} y^{equal}_{j,(c,i)}}{\beta_c - \sum_{i' \in I_c} y^{equal}_{j,(c,i')}} = f_c(\beta_c) = |I_c|.
% \end{align*}

% \ap{

\begin{align*}
    \sum_{j \in [m]} \sum_{i \in I_c} b'_{j,(c,i)} &=  \sum_{j \in [m]} \sum_{i \in I_c} \frac{(\tilde{p}_{-c})_j}{n \beta_c - |I_c| } = \sum_{j \in [m]}  \frac{|I_c| (\tilde{p}_{-c})_j}{n \beta_c - |I_c| } = f_c(\beta_c) = |I_c|.
\end{align*}

% }

Now, let $y'$ be the allocation of the Trading Post mechanism when center $c$ bids according to $b'_{j,(c,i)}$ and all other centers bid according to $\tilde{b}$. By observing that 
$u_c(y')$ is at least the utility of center $c$ for the allocation that gives each agent $i \in I_c$ exactly a $\frac{b'_{j,(c,i)}}{(\tilde{p}_{-c})_j + \sum_{i' \in I_c} b'_{j,(c,i')}}$ amount of item $j$, we can lower bound  $u_c(y')$ as follows:

% We will prove that %$y'_{j,(c,i)} = \frac{y^{equal}_{j,(c,i)}}{\beta_c}$ \ap{ 
% $y'_{j,(c,i)} = \frac{1}{ n \beta_c}$, for all $i \in I_c$, $j \in [m]$, which we show implies $u_c(y') = \frac{1}{\beta_c} u_i(y^{equal})$.

% \ap{E.g., why is the next thing equality? Maybe once the center gets the allocation, it does not choose to give so much to $(c,i)$. Things are fine because we only need a lower bound on $u_c(y')$, but strictly speaking the equality here is wrong}

% \begin{align*}
%     y'_{j,(c,i)} &= \frac{b'_{j,(c,i)}}{(\tilde{p}_{-c})_j + \sum_{i' \in I_c} b'_{j,(c,i')}}\\
%     &= \frac{\frac{y^{equal}_{j,(c,i)} (\tilde{p}_{-c})_j}{\beta_c - \sum_{i' \in I_c} y^{equal}_{j,(c,i')}}}{ (\tilde{p}_{-c})_j + \sum_{i' \in I_c}  \frac{y^{equal}_{j,(c,i')} (\tilde{p}_{-c})_j}{\beta_c - \sum_{i'' \in I_c} y^{equal}_{j,(c,i'')}}}\\
%     &=\frac{\frac{y^{equal}_{j,(c,i)} }{\beta_c - \sum_{i' \in I_c} y^{equal}_{j,(c,i')}}}{ 1 +   \frac{ \sum_{i' \in I_c} y^{equal}_{j,(c,i')}}{\beta_c - \sum_{i' \in I_c} y^{equal}_{j,(c,i')}}}\\
%     &= \frac{y^{equal}_{j,(c,i)}}{\beta_c - \sum_{i' \in I_c} y^{equal}_{j,(c,i')} + \sum_{i' \in I_c} y^{equal}_{j,(c,i')}} \\
%     &= \frac{y^{equal}_{j,(c,i)}}{\beta_c}.
% \end{align*}

% \ap{

% \begin{align*}
%     y'_{j,(c,i)} &= \frac{b'_{j,(c,i)}}{(\tilde{p}_{-c})_j + \sum_{i' \in I_c} b'_{j,(c,i')}} = \frac{ \frac{(\tilde{p}_{-c})_j}{n \beta_c - |I_c| } }{ (\tilde{p}_{-c})_j + |I_c|  \frac{(\tilde{p}_{-c})_j}{n \beta_c - |I_c| } } = \frac{1}{n \beta_c}.
% \end{align*}

\begin{align*}
u_c(y') &\geq \left( \sum_{i \in I_c}\frac{1}{|I_c|} \left( \sum_{j \in [m]} u_{j,(c,i)} \frac{b'_{j,(c,i)}}{(\tilde{p}_{-c})_j + \sum_{i' \in I_c} b'_{j,(c,i')}} \right)^{p} \right)^{1/p} \\
&= \left( \sum_{i \in I_c}\frac{1}{|I_c|} \left( \sum_{j \in [m]} u_{j,(c,i)} \frac{ \frac{(\tilde{p}_{-c})_j}{n \beta_c - |I_c| } }{ (\tilde{p}_{-c})_j + |I_c|  \frac{(\tilde{p}_{-c})_j}{n \beta_c - |I_c| } } \right)^{p} \right)^{1/p} \\
&= \left( \sum_{i \in I_c}\frac{1}{|I_c|} \left( \sum_{j \in [m]} u_{j,(c,i)} \frac{1}{n \beta_c} \right)^{p} \right)^{1/p} \\
&= \left( \sum_{i \in I_c} \frac{1}{|I_c|} \left(\frac{u_{(c,i)}(y^{equal})}{\beta_c}  \right)^{p} \right)^{1/p} \\
&= \frac{1}{\beta_c}\left( \sum_{i \in I_c} \frac{1}{|I_c|} \left( u_{(c,i)}(y^{equal}) \right)^{p} \right)^{1/p} \\
&= \frac{1}{\beta_c} u_i(y^{equal}).
\end{align*}

Finally, we can relate $u_{(c,i)}(\tilde{y})$ with $u_{(c,i)}(y^{equal})$. For $p>0$ we have
\begin{align*}
    \beta_c^{p} \sum_{i \in I_c} (u_{(c,i)}(\tilde{y}))^p   &=  \beta_c^{p}  |I_c| \sum_{i \in I_c} \frac{1}{|I_c|} (u_{(c,i)}(\tilde{y}))^p\\
    &= |I_c|\left(\beta_c \left( \sum_{i \in I_c} \frac{1}{|I_c|} (u_{(c,i)}(\tilde{y}))^p\right)^{1/p}\right)^p \\
    &= |I_c|\left(\beta_c u_c(\tilde{y})\right)^p\\
    &\ge |I_c| \left(u_c(y^{equal})\right)^p  \tag{$u_c(\tilde{y}) \ge u_c(y') = \frac{1}{\beta_c} u_c(y^{equal})$}\\
    &= |I_c| \left(\left( \sum_{i \in I_c} \frac{1}{|I_c|} (u_{c,i}(y^{equal}))^p\right)^{1/p}\right)^p\\
    &= \sum_{i \in I_c} (u_{c,i}(y^{equal}))^p.
\end{align*}

For $p < 0$ a similar argument would imply that $\beta_c^{p} \sum_{i \in I_c} (u_{(c,i)}(\tilde{y}))^p  \le  \sum_{i \in I_c} (u_{(c,i)}(y^{equal}))^p$.

\paragraph{Bounding $\beta_c$.}
We proceed to bound $\beta_c$. From the definition of the Trading Post mechanism and what we have proven thus far, we have that
$\frac{1}{n \beta_c}\sum_{j \in [m]} \sum_{i \in I_c}\left((\tilde{p}_{-c})_j + \sum_{i' \in I_c} b'_{j,(c,i')}\right) = \sum_{j \in [m]} \sum_{i \in I_c} y'_{j,(c,i)}\left((\tilde{p}_{-c})_j + \sum_{i' \in I_c} b'_{j,(c,i')}\right) = \sum_{j \in [m]} \sum_{i \in I_c} b'_{j,(c,i)} = |I_c|$,
where in the second equality we have used the fact that $y'_{j,(c,i)} = \frac{b'_{j,(c,i)}}{(\tilde{p}_{-c})_j + \sum_{i' \in I_c} b'_{j,(c,i')}}$ from the definition of the Trading Post mechanism. By rearranging we get that:
\begin{align*}
    \beta_c &= \frac{1}{n |I_c|}\sum_{j \in [m]} \sum_{i \in I_c}  \left((\tilde{p}_{-c})_j + \sum_{i' \in I_c} b'_{j,(c,i')}\right) \\
     &= \frac{1}{n |I_c|} \sum_{j \in [m]} \sum_{i \in I_c}  (\tilde{p}_{-c})_j + \frac{1}{n |I_c|} \sum_{j \in [m]} \sum_{i \in I_c}  \sum_{i' \in I_c} b'_{j,(c,i')} \\
    &= \frac{1}{n |I_c|} \sum_{j \in [m]} \sum_{i \in I_c} (\tilde{p}_{-c})_j + \frac{|I_c|}{n |I_c|} \sum_{j \in [m]} \sum_{i' \in I_c} b'_{j,(c,i')}  \\
    &= \frac{1}{n |I_c|} \sum_{j \in [m]} \sum_{i \in I_c}  (\tilde{p}_{-c})_j + \frac{|I_c|}{n} \\
    &= \frac{|I_c|}{n |I_c|} \sum_{j \in [m]} (\tilde{p}_{-c})_j + \frac{|I_c|}{n}  \\
    &\le \frac{|I_c|}{n |I_c|} \cdot n + \frac{|I_c|}{n}\\
    &\leq 1 + \frac{|I_c|}{n}. 
\end{align*}

% \ap{Let's state the upper bound as $2 \cdot \max_{c \in [C]} |I_c|/n$? } 
% \mm{Here $B=n$, $|I_c| = I_c$, so the inequality becomes $\le 1+|I_c|/n$}

% }

% \paragraph{Name this task.}
% Let $y^{equal}$ be the equal split allocation (i.e. $y^{equal}_{j,(c,i)} = \frac{1}{|N|}$). 
% \begin{align*}
%     \beta_c \le \frac{1}{|I_c|} \sum_{j \in [m]} \sum_{i \in I_c} y^{equal}_{j,(c,i)} (\tilde{p}_{-c})_j+1
%     = \frac{1}{|I_c|} \frac{1}{|N|}\sum_{j \in [m]} \sum_{i \in I_c} (\tilde{p}_{-c})_j+1
%     = \frac{1}{|I_c|} \frac{|I_c|}{|N|}\sum_{j \in [m]} (\tilde{p}_{-c})_j +1 
%     \le \frac{B}{|I_c|} \frac{|I_c|}{|N|} +1
%      = 2 
% \end{align*}

\paragraph{Putting everything together.}
Let $\beta = \max_{c \in [C]} \beta_c$. For both $p > 0$ and $p < 0$, we can establish~\Cref{eq: main bound for trading post} as follows:
\begin{align*}
    \frac{GMW_p(y^{equal})}{GMW_p(\tilde{y})} &= \left(\frac{\sum_{c \in [C]} \sum_{i \in I_c} (u_{c,i}(y^{equal}))^p}{\sum_{c \in [C]} \sum_{i \in I_c} (u_{c,i}(\tilde{y}))^p}\right)^{1/p} \\
    &\le \left(\frac{\sum_{c \in [C]} \beta_c^p \sum_{i \in I_c} (u_{c,i}(\tilde{y}))^p}{\sum_{c \in [C]} \sum_{i \in I_c} (u_{c,i}(\tilde{y}))^p}\right)^{1/p} \\
    &\le  \beta \\
    &\le 1 + \max_{c \in [C]} \frac{|I_c|}{n}.
\end{align*}

% Due to~\Cref{lemma: equal split bound} we also have that $\frac{GMW_p(y^*)}{GMW_p(y^{equal})} \le \mu_0 \, r$, where $y^* \in \argmax_{y \in \Fcal} GMW_p(y)$ is the allocation that maximizes $GMW_p(.)$ over all feasible allocations. Then,

% \[PoA_{p,\ell}(TP) =\frac{GMW_p(y^*)}{GMW_p(\tilde{y})} = \frac{GMW_p(y^*)}{GMW_p(y^{equal})}\frac{GMW_p(y^{equal})}{GMW_p(\tilde{y})} \le 2 \, \mu_0 \, r = 2 \, \ell.\]

% \paragraph{The case of $p<0$.}
% The proof for $p<0$ is identical with only two inequalities flipped. First, we can show that $\beta_c^{p} \sum_{i \in I_c} (u_{(c,i)}(\tilde{y}))^p  \le  \sum_{i \in I_c} (u_{(c,i)}(y^{equal}))^p$, which would in turn again imply that $\left(\frac{\sum_{c \in [C]} \sum_{i \in I_c} (u_{(c,i)}(y^{equal}))^p}{\sum_{c \in [C]} \sum_{i \in I_c} (u_{(c,i)}(\tilde{y}))^p}\right)^{1/p} \le \left(\frac{\sum_{c \in [C]} \beta_c^p \sum_{i \in I_c} (u_{(c,i)}(\tilde{y}))^p}{\sum_{c \in [C]} \sum_{i \in I_c} (u_{(c,i)}(\tilde{y}))^p}\right)^{1/p}$ and the rest of the proof follows. \ap{it seems like this can be incorporated into the above proof then.}

This concludes the proof of~\Cref{theorem: trading post upper bound}.
\end{proof}

\begin{theorem}\label{theorem:nsw poa}
    %Let $A$ be an $m \times |N|$, $r$-rank, and $\mu_0$-incoherent valuation matrix. Also let $TP$ be the trading post mechanism with budgets $|I_c| = |I_c|$ for all $c \in [C]$. Then 
    $PoA_{p, \ell}(\text{Trading Post}) \le 2$, for $p = 0$.
\end{theorem}

The proof of \Cref{theorem:nsw poa} is similar to the proof of~\Cref{theorem: trading post upper bound} and is deferred to~\Cref{appendix: missing}.

\subsection{Price of Anarchy of the Trading Post Mechanism: Lower Bounds}\label{section:lower bounds}

For the most well-studied Generalized $p$-Mean objectives, utilitarian social welfare ($p=1$), egalitarian social welfare ($p \rightarrow -\infty$), and Nash social welfare ($p=0$), we have shown that the Price of Anarchy of the Trading Post mechanism with proportional budgets is at most $\left( 1 + \max_{c \in [C]} \frac{|I_c|}{n} \right) \ell$, $\left( 1 + \max_{c \in [C]} \frac{|I_c|}{n} \right) \ell$, and $2$, respectively. 
Given the much better bound for $p=0$, it might seem plausible that a tighter analysis is possible for $p=1$ and/or $p \rightarrow -\infty$. In this section, we show that this is not the case. Specifically, for $p=1$ and for $p \rightarrow -\infty$, we prove that the Price of Anarchy of the Trading Post mechanism with proportional budgets is $\Omega(\sqrt{\ell})$. Although these lower bounds are not tight, they demonstrate that a dependency on $\ell = \mu_0 \, r$ is inevitable.

%Given~\Cref{lemma:welfare example incoherence} and~\Cref{lemma:maxmin example incoherence}, we are ready to prove our lower bounds.

\begin{theorem} \label{theorem: welfare poa lower bound}
     %Let $A$ be an $m \times |N|$, $r$-rank, and $\mu_0$-incoherent valuation matrix. Also let $TP$ be the trading post mechanism with budgets $B_c = |I_c|$ for all $c \in [C]$. Then 
     $PoA_{p, \ell}(\text{Trading Post}) \in  \Omega(\sqrt{\ell})$ for $p =1$.
\end{theorem}

%\ap{what is the exact statement? Given $m$ and $n$, there exists an instance with $A$ that has rank $r$ and incoherence $\mu_0$, such that ....?}

\begin{proof}
Fix an integer $\ell$. We construct an instance with $m$ items, where $m$ can be any multiple of $\ell$, and $\ell + 1$ centers, with a valuation matrix $A$ that has rank $r = \ell$ and is $\mu_0 = 1$-incoherent  (and therefore, indeed, $\mu_0 \, r = \ell$) and show that $\frac{GMW_p(y^*)}{GMW_p(\tilde{y})} \rightarrow \sqrt{\ell}$, where $\tilde{y}$ is the allocation of the Trading Post mechanism in some equilibrium and $y^*$ is an optimal solution (with respect to $GMW_p(.)$).

Pick $m$ such that $k = \frac{m}{\ell}$ is an integer (i.e., $k$ divides $m$). Our instance has $\ell + 1 = m/k+1$ centers. Center $c \in [m/k]$ represents one individual ($|I_c| = 1$), who values only items $(c-1)k+1, \cdots , ck$, and she values them equally (i.e., $v_{(c,1)} = \left[0, \cdots ,\frac{1}{\sqrt{k}}, \frac{1}{\sqrt{k}}, \cdots, 0 \right]^T$) . Center $m/k+1$ has $|I_{m/k+1}| = n- m/k$ individuals; each individual values all items equally (i.e., $\forall i \in I_{m/k+1}, v_{(m/k+1,i)} = \left[\frac{1}{\sqrt{m}}, \cdots,  \frac{1}{\sqrt{m}}\right]^T$). 

The valuation matrix of our instance is exactly the matrix in the statement of~\Cref{lemma:welfare example incoherence}. Therefore, the rank is $m/k$, and the incoherence is $\mu_0 = 1$ (and therefore, indeed, $\mu_0 \, r = \frac{m}{k}$, or $k = \frac{m}{\mu_0 \, r} = \frac{m}{\ell}$). Since we are maximizing social welfare, i.e., $GMW_1(.)$, the optimal allocation would be to allocate each item to the individual that values it the most. Therefore, $GMW_1(y^*) = \frac{m}{\sqrt{k}}$.
On the other hand, we can pick $n$ large enough so that the budget $B_{m/k + 1} = |I_{m/k+1}|$ is so large that all the items are allocated to center $m/k+1$. In this equilibrium, center $m/k + 1$ allocates each item $j$ to an individual who values it at $1/\sqrt{m}$. Thus, $GMW_1(\tilde{y}) = \sqrt{m}$. Overall:
\[\frac{GMW_p(y^*)}{GMW_p(\tilde{y})} \rightarrow \sqrt{\frac{m}{k}} = \sqrt{\mu_0 r} =\sqrt{\ell}. \qedhere \]
\end{proof}

% ========

%  We will show that for the objective when $p = 1 $:
% \[\frac{GMW_p(y^*)}{GMW_p(\tilde{y})} \rightarrow \sqrt{\mu_0 r}. \]
% \ap{what's $\tilde{y}$? what $y^*$?}
% As, $w \rightarrow \infty$ \ap{what's $w$??}.

%  Let $k = \frac{m}{\mu_0 \cdot r}$. For simplicity, let's assume that $k$ are both integers $\mu_0$ (i.e. $k$ divides $m$ perfectly). We will have $m/k+1$ centers. Center $c \in [m/k]$ has only one individual ($|I_c| = 1$) which values only items $(c-1)k+1, \cdots , ck$ equally ($v^T_{c,1} = \left[0, \cdots ,\frac{1}{\sqrt{k}}, \frac{1}{\sqrt{k}}, \cdots, 0 \right]$). center $m/k+1$ has $w = |N|- m/k$ individuals. Each individual values all items equally ($\forall i \in I_{m/k+1}, v^T_{m/k+1,i} = \left[\frac{1}{\sqrt{m}}, \cdots,  \frac{1}{\sqrt{m}}\right]$). Notice that due to \Cref{lemma:welfare example incoherence} for our valuation matrix indeed $k = \frac{m}{\mu_0 \cdot r}$. It is now easy to see that
% \[GMW_p(y^*) = \frac{m}{\sqrt{k}}\]
% On the other hand, we can see that as $w \rightarrow \infty$ all the items are going to be allocated to center $m/k+1$. In this equilibrium $GMW_p(\tilde{y}) \rightarrow \sqrt{m}$. Using the above we have that:
% \[\frac{GMW_p(y^*)}{GMW_p(\tilde{y})} \rightarrow \sqrt{\frac{m}{k}} = \sqrt{\mu_0 r} =\sqrt{\ell}\]

The following theorem shows a lower bound for $p \rightarrow -\infty$; the proof is similar to the proof of \Cref{theorem: welfare poa lower bound} and is deferred to~\Cref{appendix: missing}.

\begin{theorem} \label{theorem: maximin lower bound}
     %Let $A$ be an $m \times |N|$, $r$-rank, and $\mu_0$-incoherent valuation matrix. Also let $TP$ be the trading post mechanism with budgets $B_c = |I_c|$ for all $c \in [C]$. Then 
     $PoA_{p, \ell}(\text{Trading Post}) \in  \Omega(\sqrt{\ell})$ as $p \rightarrow -\infty$.
\end{theorem}

\section*{Acknowledgements}
The authors would like to thank Christos Boutsikas for his valuable insights into Linear Algebra.

Marios Mertzanidis and Alexandros Psomas are supported in part by an NSF CAREER award CCF-2144208, a Google AI for Social Good award, and research awards from Google and Supra.
M.\ Mertzanidis is also supported in part by NSF CCF 2209509 and 1814041.

\bibliographystyle{alpha}
\bibliography{refs}

\newcommand{\etalchar}[1]{$^{#1}$}
\begin{thebibliography}{BCKO17}

\bibitem[Ame24]{FeedingAmerica}
Feeding America.
\newblock How we end hunger.
\newblock \url{https://www.feedingamerica.org/our-work}, 2024.
\newblock Accessed: July 14, 2024.

\bibitem[BCKO17]{budish2017course}
Eric Budish, G{\'e}rard~P Cachon, Judd~B Kessler, and Abraham Othman.
\newblock Course match: A large-scale implementation of approximate competitive equilibrium from equal incomes for combinatorial allocation.
\newblock {\em Operations Research}, 65(2):314--336, 2017.

\bibitem[BEL{\etalchar{+}}07]{bennett2007kdd}
James Bennett, Charles Elkan, Bing Liu, Padhraic Smyth, and Domonkos Tikk.
\newblock Kdd cup and workshop 2007.
\newblock {\em ACM SIGKDD explorations newsletter}, 9(2):51--52, 2007.

\bibitem[BFT23]{banerjee2023Robust}
Siddhartha Banerjee, Giannis Fikioris, and Eva Tardos.
\newblock Robust pseudo-markets for reusable public resources.
\newblock In {\em Proceedings of the 24th ACM Conference on Economics and Computation}, EC '23, page 241, New York, NY, USA, 2023. Association for Computing Machinery.

\bibitem[BGM22]{branzei2022nash}
Simina Br{\^a}nzei, Vasilis Gkatzelis, and Ruta Mehta.
\newblock Nash social welfare approximation for strategic agents.
\newblock {\em Operations research}, 70(1):402--415, 2022.

\bibitem[BR11]{bhawalkar2011welfare}
Kshipra Bhawalkar and Tim Roughgarden.
\newblock Welfare guarantees for combinatorial auctions with item bidding.
\newblock In {\em Proceedings of the twenty-second annual ACM-SIAM symposium on Discrete Algorithms}, pages 700--709. SIAM, 2011.

\bibitem[Bud11]{budish2011combinatorial}
Eric Budish.
\newblock The combinatorial assignment problem: Approximate competitive equilibrium from equal incomes.
\newblock {\em Journal of Political Economy}, 119(6):1061--1103, 2011.

\bibitem[CCW23]{cai2023simultaneous}
Yang Cai, Ziyun Chen, and Jinzhao Wu.
\newblock Simultaneous auctions are approximately revenue-optimal for subadditive bidders.
\newblock In {\em 2023 IEEE 64th Annual Symposium on Foundations of Computer Science (FOCS)}, pages 134--147. IEEE, 2023.

\bibitem[Che15]{chen2015incoherence}
Yudong Chen.
\newblock Incoherence-optimal matrix completion.
\newblock {\em IEEE Transactions on Information Theory}, 61(5):2909--2923, 2015.

\bibitem[CKS16]{christodoulou2016bayesian}
George Christodoulou, Annam{\'a}ria Kov{\'a}cs, and Michael Schapira.
\newblock Bayesian combinatorial auctions.
\newblock {\em Journal of the ACM (JACM)}, 63(2):1--19, 2016.

\bibitem[CR12]{candes2012exact}
Emmanuel Candes and Benjamin Recht.
\newblock Exact matrix completion via convex optimization.
\newblock {\em Communications of the ACM}, 55(6):111--119, 2012.

\bibitem[CST15]{christodoulou2015efficiency}
George Christodoulou, Alkmini Sgouritsa, and Bo~Tang.
\newblock On the efficiency of the proportional allocation mechanism for divisible resources.
\newblock In {\em International Symposium on Algorithmic Game Theory}, pages 165--177. Springer, 2015.

\bibitem[DFL{\etalchar{+}}22]{daskalakis2022multi}
Constantinos Daskalakis, Maxwell Fishelson, Brendan Lucier, Vasilis Syrgkanis, and Santhoshini Velusamy.
\newblock Multi-item nontruthful auctions achieve good revenue.
\newblock {\em SIAM Journal on Computing}, 51(6):1796--1838, 2022.

\bibitem[FFGL20]{feldman2013simultaneous}
Michal Feldman, Hu~Fu, Nick Gravin, and Brendan Lucier.
\newblock Simultaneous auctions without complements are (almost) efficient.
\newblock {\em Games and Economic Behavior}, 123:327--341, 2020.

\bibitem[FLZ08]{feldman2008proportional}
Michal Feldman, Kevin Lai, and Li~Zhang.
\newblock The proportional-share allocation market for computational resources.
\newblock {\em IEEE Transactions on Parallel and Distributed Systems}, 20(8):1075--1088, 2008.

\bibitem[For24]{Forbes}
Forbes.
\newblock The 100 largest u.s. charities.
\newblock \url{https://www.forbes.com/top-charities/list/}, 2024.
\newblock Accessed: July 14, 2024.

\bibitem[GBI21a]{gorokh2021monetary}
Artur Gorokh, Siddhartha Banerjee, and Krishnamurthy Iyer.
\newblock From monetary to nonmonetary mechanism design via artificial currencies.
\newblock {\em Mathematics of Operations Research}, 46(3):835--855, 2021.

\bibitem[GBI21b]{gorokh2021remarkable}
Artur Gorokh, Siddhartha Banerjee, and Krishnamurthy Iyer.
\newblock The remarkable robustness of the repeated fisher market.
\newblock In {\em Proceedings of the 22nd ACM Conference on Economics and Computation}, EC '21, page 562, New York, NY, USA, 2021. Association for Computing Machinery.

\bibitem[GPT21]{gkatzelis2021fair}
Vasilis Gkatzelis, Alexandros Psomas, and Xizhi Tan.
\newblock Fair and efficient online allocations with normalized valuations.
\newblock In {\em Proceedings of the AAAI conference on artificial intelligence}, volume~35, pages 5440--5447, 2021.

\bibitem[IFO15]{isinkaye2015recommendation}
Folasade~Olubusola Isinkaye, Yetunde~O Folajimi, and Bolande~Adefowoke Ojokoh.
\newblock Recommendation systems: Principles, methods and evaluation.
\newblock {\em Egyptian informatics journal}, 16(3):261--273, 2015.

\bibitem[JT04]{johari2004efficiency}
Ramesh Johari and John~N Tsitsiklis.
\newblock Efficiency loss in a network resource allocation game.
\newblock {\em Mathematics of Operations Research}, 29(3):407--435, 2004.

\bibitem[KLL{\etalchar{+}}23]{kelner2023matrix}
Jonathan~A Kelner, Jerry Li, Allen Liu, Aaron Sidford, and Kevin Tian.
\newblock Matrix completion in almost-verification time.
\newblock In {\em 2023 IEEE 64th Annual Symposium on Foundations of Computer Science (FOCS)}, pages 2102--2128. IEEE, 2023.

\bibitem[KP09]{koutsoupias2009worst}
Elias Koutsoupias and Christos Papadimitriou.
\newblock Worst-case equilibria.
\newblock {\em Computer science review}, 3(2):65--69, 2009.

\bibitem[Pre22]{prendergast2022allocation}
Canice Prendergast.
\newblock The allocation of food to food banks.
\newblock {\em Journal of Political Economy}, 130(8):1993--2017, 2022.

\bibitem[SS77]{shapley1977trade}
Lloyd Shapley and Martin Shubik.
\newblock Trade using one commodity as a means of payment.
\newblock {\em Journal of political economy}, 85(5):937--968, 1977.

\bibitem[Wal14]{walsh2014allocation}
Toby Walsh.
\newblock Allocation in practice.
\newblock In {\em Joint German/Austrian Conference on Artificial Intelligence (K{\"u}nstliche Intelligenz)}, pages 13--24. Springer, 2014.

\bibitem[Zha05]{zhang2005efficiency}
Li~Zhang.
\newblock The efficiency and fairness of a fixed budget resource allocation game.
\newblock In {\em International Colloquium on Automata, Languages, and Programming}, pages 485--496. Springer, 2005.

\end{thebibliography}

\appendix

\section{Missing Proofs}\label{appendix: missing}
\begin{proof}[Proof of~\Cref{lemma:welfare example incoherence}]
    It is straightforward to check that $A$ has rank $m/k$. Notice that by the definition of $\mu_0$-incoherence (\Cref{definition:incoherence}), we simply need to show that $\max_{i \in [m]} \norm{U^T e_i}^2_2 = 1/k$. To do so we calculate the singular value decomposition of $A$. We start by calculating the singular values of $A$. Since our matrix has rank $m/k$, we know that we have $m/k$ non-zero singular values which are equal to the square roots of the non-zero eigenvalues of $A^TA$. To find those eigenvalues we need to solve the equation $det(A^TA - I\lambda) = 0$. It is easy to check that apart from $0$, $1$ is also an eigenvalue. For ease of notation, let $w = n-m/k$.

\begin{align*}
    det(A^TA - I\lambda) &= \begin{vmatrix}
1-\lambda & 0 & 0 & \cdots & \sqrt{\frac{k}{m}} & \sqrt{\frac{k}{m}} & \cdots & \sqrt{\frac{k}{m}} \\
0 & 1-\lambda & 0 & \cdots & \sqrt{\frac{k}{m}} & \sqrt{\frac{k}{m}}& \cdots & \sqrt{\frac{k}{m}} \\
0 & 0 & 1-\lambda & \cdots & \sqrt{\frac{k}{m}} & \sqrt{\frac{k}{m}}& \cdots & \sqrt{\frac{k}{m}} \\
\vdots & \vdots & \vdots & \ddots & \vdots & \vdots & \ddots & \vdots \\
\sqrt{\frac{k}{m}} & \sqrt{\frac{k}{m}} & \sqrt{\frac{k}{m}} & \cdots  & 1- \lambda &1 & \cdots &1 \\
\sqrt{\frac{k}{m}} & \sqrt{\frac{k}{m}} & \sqrt{\frac{k}{m}} & \cdots  & 1 &1- \lambda & \cdots &1 \\
\vdots & \vdots & \vdots & \ddots & \vdots & \vdots & \ddots & \vdots \\
\sqrt{\frac{k}{m}} & \sqrt{\frac{k}{m}} & \sqrt{\frac{k}{m}} & \cdots  & 1 & 1 &\cdots &1 - \lambda \\
\end{vmatrix} \\
&=\begin{vmatrix}
1-\lambda & 0 & 0 & \cdots & \sqrt{\frac{k}{m}} & \sqrt{\frac{k}{m}} & \cdots & \sqrt{\frac{k}{m}} \\
\lambda - 1 & 1-\lambda & 0 & \cdots & 0 & 0 & \cdots & 0 \\
\lambda - 1 & 0 & 1-\lambda & \cdots & 0 & 0& \cdots & 0 \\
\vdots & \vdots & \vdots & \ddots & \vdots & \vdots & \ddots & \vdots \\
\sqrt{\frac{k}{m}} & \sqrt{\frac{k}{m}} & \sqrt{\frac{k}{m}} & \cdots  & 1- \lambda &1 & \cdots &1 \\
\sqrt{\frac{k}{m}} & \sqrt{\frac{k}{m}} & \sqrt{\frac{k}{m}} & \cdots  & 1 &1- \lambda & \cdots &1 \\
\vdots & \vdots & \vdots & \ddots & \vdots & \vdots & \ddots & \vdots \\
\sqrt{\frac{k}{m}} & \sqrt{\frac{k}{m}} & \sqrt{\frac{k}{m}} & \cdots  & 1 & 1 &\cdots &1 - \lambda \\
\end{vmatrix} \tag{Subtracting row $1$ from rows $2$ to $m/k$}\\
&=\begin{vmatrix}
1-\lambda & 0 & 0 & \cdots & \sqrt{\frac{k}{m}} & \sqrt{\frac{k}{m}} & \cdots & \sqrt{\frac{k}{m}} \\
\lambda - 1 & 1-\lambda & 0 & \cdots & 0 & 0 & \cdots & 0 \\
\lambda - 1 & 0 & 1-\lambda & \cdots & 0 & 0& \cdots & 0 \\
\vdots & \vdots & \vdots & \ddots & \vdots & \vdots & \ddots & \vdots \\
\sqrt{\frac{k}{m}} & \sqrt{\frac{k}{m}} & \sqrt{\frac{k}{m}} & \cdots  & 1- \lambda &1 & \cdots &1 \\
0 & 0 & 0 & \cdots  & \lambda &- \lambda & \cdots &0 \\
\vdots & \vdots & \vdots & \ddots & \vdots & \vdots & \ddots & \vdots \\
0 & 0 & 0 & \cdots  & \lambda & 0 &\cdots & - \lambda \\
\end{vmatrix} \tag{Subtracting row $m/k+1$ from rows $m/k+2$ to $|N|$}\\
&=\begin{vmatrix}
1-\lambda & 0 & 0 & \cdots & \sqrt{\frac{k}{m}} & \sqrt{\frac{k}{m}} & \cdots & \sqrt{\frac{k}{m}} \\
\lambda - 1 & 1-\lambda & 0 & \cdots & 0 & 0 & \cdots & 0 \\
\lambda - 1 & 0 & 1-\lambda & \cdots & 0 & 0& \cdots & 0 \\
\vdots & \vdots & \vdots & \ddots & \vdots & \vdots & \ddots & \vdots \\
\sqrt{\frac{m}{k}} & 0 & 0 & \cdots  & 1- \lambda &1 & \cdots &1 \\
0 & 0 & 0 & \cdots  & \lambda &- \lambda & \cdots &0 \\
\vdots & \vdots & \vdots & \ddots & \vdots & \vdots & \ddots & \vdots \\
0 & 0 & 0 & \cdots  & \lambda & 0 &\cdots & - \lambda \\
\end{vmatrix} \tag{Subtracting rows $2$ to $m/k$ times $\frac{1}{1-\lambda} \sqrt{\frac{k}{m}}$ from row $m/k+1$, assuming $\lambda \neq 1$}\\
&=\begin{vmatrix}
1-\lambda & 0 & 0 & \cdots & w\sqrt{\frac{k}{m}} & 0 & \cdots & 0 \\
\lambda - 1 & 1-\lambda & 0 & \cdots & 0 & 0 & \cdots & 0 \\
\lambda - 1 & 0 & 1-\lambda & \cdots & 0 & 0& \cdots & 0 \\
\vdots & \vdots & \vdots & \ddots & \vdots & \vdots & \ddots & \vdots \\
\sqrt{\frac{m}{k}} & 0 & 0 & \cdots  & 1- \lambda &1 & \cdots &1 \\
0 & 0 & 0 & \cdots  & \lambda &- \lambda & \cdots &0 \\
\vdots & \vdots & \vdots & \ddots & \vdots & \vdots & \ddots & \vdots \\
0 & 0 & 0 & \cdots  & \lambda & 0 &\cdots & - \lambda \\
\end{vmatrix} \tag{Adding rows $m/k+2$ to $|N|$ times $\frac{1}{\lambda} \sqrt{\frac{k}{m}}$ to row $1$, assuming $\lambda \neq 0$}\\
&=\begin{vmatrix}
1-\lambda & 0 & 0 & \cdots & w\sqrt{\frac{k}{m}} & 0 & \cdots & 0 \\
\lambda - 1 & 1-\lambda & 0 & \cdots & 0 & 0 & \cdots & 0 \\
\lambda - 1 & 0 & 1-\lambda & \cdots & 0 & 0& \cdots & 0 \\
\vdots & \vdots & \vdots & \ddots & \vdots & \vdots & \ddots & \vdots \\
\sqrt{\frac{m}{k}} & 0 & 0 & \cdots  & w- \lambda &0 & \cdots & \\
0 & 0 & 0 & \cdots  & \lambda &- \lambda & \cdots &0 \\
\vdots & \vdots & \vdots & \ddots & \vdots & \vdots & \ddots & \vdots \\
0 & 0 & 0 & \cdots  & \lambda & 0 &\cdots & - \lambda \\
\end{vmatrix} \tag{Adding rows $m/k+2$ to $|N|$ times $\frac{1}{\lambda}$ to row $1$, assuming $\lambda \neq 0$}\\
&=\begin{vmatrix}
1-\lambda - \frac{w}{w-\lambda} & 0 & 0 & \cdots & 0 & 0 & \cdots & 0 \\
\lambda - 1 & 1-\lambda & 0 & \cdots & 0 & 0 & \cdots & 0 \\
\lambda - 1 & 0 & 1-\lambda & \cdots & 0 & 0& \cdots & 0 \\
\vdots & \vdots & \vdots & \ddots & \vdots & \vdots & \ddots & \vdots \\
\sqrt{\frac{m}{k}} & 0 & 0 & \cdots  & w- \lambda &0 & \cdots & \\
0 & 0 & 0 & \cdots  & \lambda &- \lambda & \cdots &0 \\
\vdots & \vdots & \vdots & \ddots & \vdots & \vdots & \ddots & \vdots \\
0 & 0 & 0 & \cdots  & \lambda & 0 &\cdots & - \lambda \\
\end{vmatrix}. \tag{Subtracting row $m/k+1$ times $\frac{w}{w-\lambda}\sqrt{\frac{k}{m}}$ from row $1$, assuming $\lambda \neq 0$, assuming $\lambda \neq w$}\\
\end{align*}

Thus our characteristic polynomial reduces to $1-\lambda - \frac{w}{w-\lambda} = 0$, which has solutions $0$ and $w+1$. However we have assumed that $\lambda \neq 0$ and thus only $w+1$ is a a valid eigenvalue from this calculation. We got one non-zero eigenvalue. However, due to the fact that $A$ has rank $m/k$ we know that we should have an extra $m/k-1$ non-zero eigenvalues. In our calculations we have only assumed that $\lambda \neq w$ and $\lambda \neq 1$. It is easy to check that $\lambda = w$ is not an eigenvalue of our matrix thus the rest non-zero eigenvalues must be equal to $1$.

Now that we have successfully found the singular values of $A$ we can find $U$ by finding the eigenvectors of $AA^T$.

\begin{align*}
   \left( AA^T - I \lambda\right)v &= \begin{bmatrix}
\frac{1}{k}+\frac{w}{m} - \lambda & \frac{1}{k}+\frac{w}{m} & \cdots & \frac{w}{m} & \frac{w}{m} & \cdots & \frac{w}{m} \\
\frac{1}{k}+\frac{w}{m} & \frac{1}{k}+\frac{w}{m} - \lambda & \cdots & \frac{w}{m} & \frac{w}{m} & \cdots & \frac{w}{m} \\
\vdots & \vdots & \ddots & \vdots & \vdots & \ddots & \vdots \\
\frac{w}{m} & \frac{w}{m} & \cdots & \frac{1}{k}+\frac{w}{m} -\lambda & \frac{1}{k}+\frac{w}{m} & \cdots & \frac{w}{m}\\
\vdots & \vdots & \ddots & \vdots & \vdots & \ddots & \vdots \\
\frac{w}{m} & \frac{w}{m} & \cdots &\frac{w}{m} & \frac{w}{m} & \cdots & \frac{1}{k}+\frac{w}{m} -\lambda
\end{bmatrix} \\
&= \begin{bmatrix}
\frac{1}{k}+\frac{w}{m} - \lambda & \frac{1}{k}+\frac{w}{m} & \cdots & \frac{w}{m} & \frac{w}{m} & \cdots & \frac{w}{m} \\
\lambda &  - \lambda & \cdots & 0 & 0 & \cdots & 0 \\
\vdots & \vdots & \ddots & \vdots & \vdots & \ddots & \vdots \\
\frac{w}{m} & \frac{w}{m} & \cdots & \frac{1}{k}+\frac{w}{m} -\lambda & \frac{1}{k}+\frac{w}{m} & \cdots & \frac{w}{m}\\
\vdots & \vdots & \ddots & \vdots & \vdots & \ddots & \vdots \\
0 & 0 & \cdots & 0 & 0 & \cdots & -\lambda
\end{bmatrix}. \tag{Subtracting row $ (i-1) \cdot k+1, i \in [m/k]$ from rows $(i-1) \cdot k+1$ to  $i \cdot k$}
\end{align*}
From the above we can infer that for $\lambda \neq 0$ it must be the case that $v_i = v_j$, for $i,j \in [(z-1)\cdot k +1, z \cdot k], z \in [m/k]$. For ease of analysis, rename $x_i = v_{(i-1) \cdot k+1}$, $i \in [m/k]$. We have the following conditions:
\begin{align*}
    (1-\lambda)x_i + k\frac{w}{m} \sum_{j \in [m/k]} x_j = 0, \forall i \in [m/k].
\end{align*}
For $\lambda  = w+1$ the eigenvector is  the $x_i = \frac{1}{\sqrt{m}}, \forall i \in [m/k]$. Now for $\lambda = 1$ we need to find $m/k-1$ orthonormal eigenvectors. Consider the vectors such that $x_z = 0, z \in \{1, \cdots, i-1\}$, $x_i = - \sqrt{\frac{m-ik}{k(m-(i-1)k)}}$, and $x_j = \sqrt{\frac{k}{(m-ik)(m-(i-1)k)}}, j \in \{i+1, \cdots, m/k \}$, for $i \in [m/k-1]$. By construction, $v_i = v_j$, for $i,j \in [(z-1)\cdot k +1, z \cdot k], z \in [m/k]$. It is also easy to check that $\sum_{j \in [m/k]} x_j = 0$. The above two conditions imply that our vectors are indeed eigenvectors of our matrix for eigenvalue $\lambda = 1$. Finally, it is easy to check that the inner product of any two of these vectors is 0 and they are normal. Thus we have successfully computed $U$:
\begin{equation*}
    U = \begin{bmatrix}
        \frac{1}{\sqrt{m}} & - \sqrt{\frac{m-k}{km}} & 0 & \cdots & 0 & \cdots & 0\\
        \frac{1}{\sqrt{m}} & - \sqrt{\frac{m-k}{km}} & 0 & \cdots & 0 & \cdots & 0\\
        \vdots & \vdots & \vdots & \ddots & \vdots & \ddots & \vdots \\
        \frac{1}{\sqrt{m}} & \sqrt{\frac{k}{m-k}} &  - \sqrt{\frac{m-2k}{k(m-k)}} & \cdots & 0 &\cdots & 0 \\
        \vdots & \vdots & \vdots & \ddots & \vdots & \ddots & \vdots \\
        \frac{1}{\sqrt{m}} & \sqrt{\frac{k}{m-k}} &  \sqrt{\frac{k}{(m-2k)(m-k)}} & \cdots & - \sqrt{\frac{m-ik}{k(m-(i-1)k)}} &\cdots & 0  \\
        \vdots & \vdots & \vdots & \ddots & \vdots & \ddots & \vdots \\
         \frac{1}{\sqrt{m}} & \sqrt{\frac{k}{m-k}} &  \sqrt{\frac{k}{(m-2k)(m-k)}} & \cdots & \sqrt{\frac{k}{(m-ik)(m-(i-1)k)}} &\cdots & \frac{1}{\sqrt{2k}} 
    \end{bmatrix}.
\end{equation*}

To compute $\norm{U^T e_i}^2_2$ for any $i \in [m]$, we have:
\begin{align*}
    \norm{U^T e_i}^2_2 &= \frac{1}{m} + \sum_{ j \in \left[\floor{i/k}-1\right]} \frac{k}{(m-jk)(m-((j-1)k)} +  \frac{m-(\floor{i/k})k}{k(m-(\floor{i/k}-1)k)} \\
    &= \frac{1}{m} + \sum_{ j \in \left[\floor{i/k}-2\right]} \frac{k}{(m-jk)(m-((j-1)k)} + \frac{k}{(m-(\floor{i/k}-1)k)(m-((\floor{i/k}-2)k)} \\
    &\quad +  \frac{m-(\floor{i/k})k}{k(m-(\floor{i/k}-1)k)}\\
    &= \frac{1}{m} + \sum_{ j \in \left[\floor{i/k}-2\right]} \frac{k}{(m-jk)(m-((j-1)k)} + \frac{k^2 + (m-(\floor{i/k})k)(m-((\floor{i/k}-2)k)}{k(m-(\floor{i/k}-1)k)(m-((\floor{i/k}-2)k)}\\
    &= \frac{1}{m} + \sum_{ j \in \left[\floor{i/k}-2\right]} \frac{k}{(m-jk)(m-((j-1)k)} \\
    &\quad + \frac{k^2 + \left((m-(\floor{i/k}-1)k)-k\right)\left((m-((\floor{i/k}-1)k)+k\right)}{k(m-(\floor{i/k}-1)k)(m-((\floor{i/k}-2)k)}\\
    &= \frac{1}{m} + \sum_{ j \in \left[\floor{i/k}-2\right]} \frac{k}{(m-jk)(m-((j-1)k)} 
    + \frac{k^2 + \left(m-(\floor{i/k}-1)k\right)^2- k^2}{k(m-(\floor{i/k}-1)k)(m-((\floor{i/k}-2)k)}\\
    &= \frac{1}{m} + \sum_{ j \in \left[\floor{i/k}-2\right]} \frac{k}{(m-jk)(m-((j-1)k)} + \frac{\left(m-(\floor{i/k}-1)k\right)}{k(m-((\floor{i/k}-2)k)}\\
    &= \norm{U^T e_{i-k}}^2_2 \\
    &= \norm{U^T e_{1}}^2_2 \\
    &= \frac{1}{m} + \frac{m-k}{mk} \\
    &= \frac{1}{k}.
\end{align*}
\end{proof}

\begin{proof}[Proof of \Cref{lemma:maxmin example incoherence}]
    For ease of analysis, we will assume that $k$ perfectly divides $m$. it is straight forward to check that $A$ is $2$-rank. Notice that by the definition of $\mu_0$-incoherence (\Cref{definition:incoherence}), we simply need to show that $\max_{i \in [m]} \norm{U^T e_i}^2_2 = 1/k$. To do so, we first calculate the singular values of $A$. Since our matrix has rank 2, we know that we have 2 non-zero singular values which are equal to the square root of the non-zero eigenvalues of $A^TA$. To find those eigenvalues we need to solve the equation $det(A^TA - I\lambda) = 0$. 

\begin{align*}
    det(A^TA - I\lambda) &= \begin{vmatrix}
1-\lambda & \sqrt{\frac{k}{m}} & \cdots & \sqrt{\frac{k}{m}}\\
\sqrt{\frac{k}{m}} & 1 -\lambda& \cdots & 1\\
\vdots & \vdots & \ddots & \vdots \\
\sqrt{\frac{k}{m}} & 1 & \cdots & 1-\lambda
\end{vmatrix} \\
&=\begin{vmatrix}
1-\lambda & \sqrt{\frac{k}{m}} & \sqrt{\frac{k}{m}} & \cdots & \sqrt{\frac{k}{m}} &\sqrt{\frac{k}{m}}\\
0 &  -\lambda&0& \cdots & 0 &\lambda\\
\vdots & \vdots& \vdots & \ddots & \vdots& \vdots \\
0 &  0&0& \cdots & -\lambda &\lambda\\
\sqrt{\frac{k}{m}} & 1 & 1 & \cdots & 1 & 1-\lambda
\end{vmatrix} \tag{Subtracting row $n$ from rows $2$ to $n-1$}\\
&=\begin{vmatrix}
1-\lambda & 0 & 0 & \cdots & 0 & (n-1)\sqrt{\frac{k}{m}}\\
0 &  -\lambda&0& \cdots & 0 &\lambda\\
\vdots & \vdots& \vdots & \ddots & \vdots& \vdots \\
0 &  0&0& \cdots & -\lambda &\lambda\\
\sqrt{\frac{k}{m}} & 1 & 1 & \cdots & 1 & 1-\lambda
\end{vmatrix} \tag{Adding rows $2$ to $n-1$ times $\frac{1}{\lambda} \sqrt{\frac{k}{m}}$ from row $1$, assuming $\lambda \neq 0$}\\
&=\begin{vmatrix}
1-\lambda & 0 & 0 & \cdots & 0 & (n-1)\sqrt{\frac{k}{m}}\\
0 &  -\lambda&0& \cdots & 0 &\lambda\\
\vdots & \vdots& \vdots & \ddots & \vdots& \vdots \\
0 &  0&0& \cdots & -\lambda &\lambda\\
\sqrt{\frac{k}{m}} & 0 & 0 & \cdots & 0 & n-1-\lambda
\end{vmatrix} \tag{Adding rows $2$ to $n-1$ times $\frac{1}{\lambda}$ from row $n$, assuming $\lambda \neq 0$}\\
&=\begin{vmatrix}
1-\lambda & 0 & 0 & \cdots & 0 & (n-1)\sqrt{\frac{k}{m}}\\
0 &  -\lambda&0& \cdots & 0 &\lambda\\
\vdots & \vdots& \vdots & \ddots & \vdots& \vdots \\
0 &  0&0& \cdots & -\lambda &\lambda\\
0 & 0 & 0 & \cdots & 0 & n-1-\lambda - (n-1)\frac{k}{m}\frac{1}{1-\lambda}
\end{vmatrix} \tag{Subtracting rows $1$ times $\frac{1}{1-\lambda}\sqrt{\frac{k}{m}}$ from row $n$, assuming $\lambda \neq 1$}\\
&= (1-\lambda)(-\lambda)^{n-1}\left(n-1-\lambda - (n-1)\frac{k}{m}\frac{1}{1-\lambda}\right). \tag{Determinant of triangular matrix}
\end{align*}
Since we have assumed that $\lambda \neq 1$ and $\lambda \neq 0$, our equation $det(A^TA - I\lambda) = 0$ is equivalent to $n-1-\lambda - (n-1)\frac{k}{m}\frac{1}{1-\lambda} = 0$. The solutions to this equation are:
\[\lambda_{1,2} = \frac{n \pm \sqrt{n^2-4(n-1)\left(1- \frac{k}{m}\right)}}{2}\]
One useful property of these $\lambda$ is that $\left((\lambda_1-1)(\lambda_2-1)\right)^2(m-k)^2 = k^2(\lambda_1 \lambda_2)^2$. To see this we have:
\begin{align*}
    \left((\lambda_1-1)(\lambda_2-1)\right)^2(m-k)^2 & = \left((m-k)(\lambda_1 \lambda_2 -(\lambda_1 + \lambda_2)+1)\right)^2 \\
    &= \left((m-k)\left((n-1)\left(1-\frac{k}{m}\right) -n+1\right)\right)^2 \\
    &= \left((m-k)(n-1)\frac{k}{m}\right)^2 \\
    &= k^2 \left(\left(1-\frac{k}{m}\right)(n-1)\right)^2\\
    &= k^2 (\lambda_1 \lambda_2)^2
\end{align*}

Now that we have calculated the singular values of our matrix we can calculate matrix $U$. Since $A$ has rank $2$ we will only need the two leftmost vectors of $U$. These will be the non-trivial, unit-length vectors that are in the null-space of $AA^T - I \lambda_{1,2}$. For a vector $v$ such that $\left( AA^T - I \lambda_{1}\right)v = 0$ we have that:

\begin{align*}
   \left( AA^T - I \lambda_{1}\right)v &= \begin{bmatrix}
\frac{n-1}{m}+\frac{1}{k}-\lambda & \frac{n-1}{m}+\frac{1}{k} & \cdots &\frac{n-1}{m} & \cdots & \frac{n-1}{m}\\
\frac{n-1}{m}+\frac{1}{k} & \frac{n-1}{m}+\frac{1}{k} -\lambda& \cdots& \frac{n-1}{m} & \cdots & \frac{n-1}{m}\\
\vdots & \vdots & \ddots & \vdots & \ddots & \vdots \\
\frac{n-1}{m} & \frac{n-1}{m}& \cdots& \frac{n-1}{m} -\lambda & \cdots & \frac{n-1}{m}\\
\vdots & \vdots & \ddots & \vdots & \ddots & \vdots \\
\frac{n-1}{m} & \frac{n-1}{m}& \cdots& \frac{n-1}{m} & \cdots & \frac{n-1}{m}-\lambda\\
\end{bmatrix} v \\
&= \begin{bmatrix}
\frac{1}{k}-\lambda & \frac{1}{k} & \cdots & 0 & \cdots &\lambda\\
\frac{1}{k} & \frac{1}{k} -\lambda& \cdots& 0 & \cdots & \lambda\\
\vdots & \vdots & \ddots & \vdots & \ddots & \vdots \\
0 & 0 & \cdots& -\lambda & \cdots & \lambda\\
\vdots & \vdots & \ddots & \vdots & \ddots & \vdots \\
\frac{n-1}{m} & \frac{n-1}{m}& \cdots& \frac{n-1}{m} & \cdots & \frac{n-1}{m}-\lambda\\
\end{bmatrix} v  \tag{Subtracting row $m$ from rows $1$ to $m-1$}\\
&= \begin{bmatrix}
\frac{1}{k}-\lambda & \frac{1}{k} & \cdots & \frac{1}{k} & 0 & \cdots &\lambda\\
\lambda & -\lambda& \cdots& 0 & 0 & \cdots & 0\\
\vdots & \vdots & \ddots & \vdots & \vdots & \ddots & \vdots \\
\lambda & 0& \cdots& -\lambda & 0 & \cdots & 0\\
0 & 0 & \cdots& 0 &-\lambda & \cdots & \lambda\\
\vdots & \vdots & \ddots & \vdots & \vdots & \ddots & \vdots \\
\frac{n-1}{m} & \frac{n-1}{m}& \cdots& \frac{n-1}{m} &\frac{n-1}{m} & \cdots & \frac{n-1}{m}-\lambda\\
\end{bmatrix} v  \tag{Subtracting row $1$ from rows $2$ to $k$}\\
\end{align*}
Based on the above we can see that due to rows $2$ to $k$ we have that $v_i = v_j$ for $1 \le i,j \le k$. Equivalently, due to rows $k+1$ to $m-1$ we have that $v_i = v_j$ for $k+1 \le i,j \le m-1$. From row $1$ and our previous observations we have that $(1-\lambda)v_1 +\lambda v_{m} = 0$. We also know that $\norm{v}_2=1$ and thus $k v_1^2+ (m-k) v_m^{2} =1$. Combining the two equations we get that:
\begin{align*}
    v_1 &= \frac{\lambda}{\sqrt{k\lambda^2 +(\lambda-1)^2(m-k)}} \\
    v_m &= \frac{\lambda-1}{\sqrt{k\lambda^2 +(\lambda-1)^2(m-k)}}
\end{align*}
Since our singular values are positive reals we have that $v_1^2 \ge v_m^2$. This means that 
\begin{align*}
    &\max_{i \in [m]} \norm{U^T e_i}^2_2 = \norm{U^T e_1}^2_2 \\
    &= \frac{\lambda_1^2}{k\lambda_1^2 +(\lambda_1-1)^2(m-k)}+\frac{\lambda_2^2}{k\lambda_2^2 +(\lambda_2-1)^2(m-k)}\\
    &=\frac{\lambda_1^2 \left( k \lambda_2^2 + (\lambda_2-1)^2(m-k) \right)+ \lambda_2^2 \left( k \lambda_1^2 + (\lambda_1-1)^2(m-k) \right)}{\left (k\lambda_1^2 +(\lambda_1-1)^2(m-k) \right) \left( k\lambda_2^2 +(\lambda_2-1)^2(m-k)\right)}\\
    &=\frac{2k (\lambda_1\lambda_2)^2 + (\lambda_1\lambda_2-\lambda_1)^2(m-k) + (\lambda_2\lambda_1-\lambda_2)^2(m-k) }{k^2 (\lambda_1\lambda_2)^2 + k(\lambda_1\lambda_2-\lambda_1)^2(m-k) + k(\lambda_2\lambda_1-\lambda_2)^2(m-k)+ \left((\lambda_1-1)(\lambda_2-1)\right)^2(m-k)^2}\\
    &=\frac{2k (\lambda_1\lambda_2)^2 + (\lambda_1\lambda_2-\lambda_1)^2(m-k) + (\lambda_2\lambda_1-\lambda_2)^2(m-k) }{2k^2 (\lambda_1\lambda_2)^2 + k(\lambda_1\lambda_2-\lambda_1)^2(m-k) + k(\lambda_2\lambda_1-\lambda_2)^2(m-k)}\\
    &= \frac{1}{k}.
\end{align*}
\end{proof}

\begin{proof}[Proof of~\Cref{thm: first price lower bound p is 1}]
Fix an integer $\ell$, and pick the number of items $m$ such that $\ell$ perfectly divides $m$. Let $k = m/\ell$. We construct an instance with $C=m$ centers, where each center serves $\ell = m/k$ individuals. Without loss of generality let $g^{\Mcal}_1(\ell, \cdots, \ell) \in \argmin_{c \in [C]} g^{\Mcal}_c(\ell, \cdots, \ell)$. 

For center one, each individual $i \in I_1$ values only items $(i-1)k+1, \cdots , ik$, and she values them equally (i.e., $v_{(1,i)} = \left[0, \cdots ,\frac{1}{\sqrt{k}}, \frac{1}{\sqrt{k}}, \cdots, 0 \right]^T$). For all remaining centers $c \in [C]\setminus \{1\}$, each individual $i \in I_c$, values all items equally (i.e. $v_{(c,i)} = \left[\frac{1}{\sqrt{m}}, \cdots , \frac{1}{\sqrt{m}} \right]^T$). The induced valuation matrix is precisely the valuation matrix of~\Cref{lemma:welfare example incoherence}; therefore its rank is exactly $r=m/k$ and its incoherence is $\mu_0 = 1$, and thus, indeed, $\ell = m/k$, for the induced valuation matrix. 

In the optimal allocation, center one receives all items, for a welfare of $m/\sqrt{k}$. Now, assume that there exists an equilibrium where the sum of fractions of items center one receives is strictly more than one (based on the tie-breaking rule, centers that are tied for the highest bid can receive fractional allocations). Then, there must exist a center $c^* \in [C] \setminus \{ 1 \}$ that receives a sum of fractions that is strictly less than one. Also, it must be the case that center one has a strictly positive bid on at least two items (by the ``zero bids imply no allocation'' assumption); let $j$ be an item that center one receives (potentially fractionally). Center $c^*$ could bid $g^{\Mcal}_c(\ell, \cdots, \ell) \geq g^{\Mcal}_1(\ell, \cdots, \ell) > b_{(j,1)}$ for item $j$ and receive it entirely, contradicting the equilibrium property.
This implies that at any equilibrium, the sum of fractions of items that center one receives is at most one, and thus the maximum welfare attained at an equilibrium is at most $\sqrt{m} +\frac{1}{\sqrt{k}} - \frac{1}{\sqrt{m}}$. Therefore, the Price of Anarchy is at least  $\frac{m/\sqrt{k}}{\sqrt{m} +\frac{1}{\sqrt{k}} - \frac{1}{\sqrt{m}}} \in \Omega(\sqrt{\frac{m}{k}}) = \Omega(\sqrt{\ell})$. 
\end{proof}

\begin{proof}[Proof of~\Cref{lemma: technical bound}]
     It suffices to prove that
     $\hat{y} \in \argmax_{y: \sum_{i \in [k]} y_i \le m} \left(\sum_{i \in [k]} (y_i)^p \right)^{1/p}$, where $\hat{y}_i = \frac{m}{k}$, since \begin{equation*}
        \left( \sum_{i \in [k]} (\hat{y}_i)^p\right)^{1/p} =  \left( \sum_{i \in [k]} \left(\frac{m}{k}\right)^p\right)^{1/p}= \frac{m}{k^{1-1/p}}.
    \end{equation*} 
    We will prove that for any $\tilde{y}$ that is not equal to $\hat{y}$ pointwise, there is ``local improvement,'' with respect to the objective $\left( \sum_{i \in [k]} \left( y_{i}\right)^{p} \right)^{1/p}$. Consider any $\tilde{y}$ such that there exist $i', i'' \in [k]$, such that $\tilde{y}_{i'} - \tilde{y}_{i''} > \epsilon$. Let $z$ be a vector such that $\forall i \in [k]: i \neq i'', i \neq i'$, $z_i = \tilde{y}_i$, $z_{i'} = \tilde{y}_{i'} -\epsilon/2$, and $z_{i''} = \tilde{y}_{i''} +\epsilon/2$. Then, we have that
    \begin{align*}
        \sum_{i \in [k]} (z_i)^p - \sum_{i \in [k]} (\tilde{y}_i)^p &= (\tilde{y}_{i''} +\epsilon/2)^p - (\tilde{y}_{i''})^p + ( \tilde{y}_{i'} -\epsilon/2)^p - (\tilde{y}_{i'})^p  \\ 
        &= \epsilon/2 \left( \frac{(\tilde{y}_{i''} +\epsilon/2)^p - (\tilde{y}_{i''})^p}{\epsilon/2} \right) - \epsilon/2 \left( \frac{(\tilde{y}_{i'})^p- (\tilde{y}_{i'} -\epsilon/2)^p}{\epsilon/2} \right).
    \end{align*}
   The \textit{Mean Value Theorem}, states that for any function $f: \R \rightarrow \R$ that is continuous on $[\alpha, \beta]$ and differentiable on $(\alpha, \beta)$, there exists a point $c \in (\alpha, \beta)$ such that $\frac{f(x)}{dx} \Bigr|_{x = c} = \frac{f(\beta) - f(\alpha)}{\beta - \alpha}$. Therefore, using the Mean Value Theorem for the function $f(x) = x^p$, and $p \in (0,1]$, we have
        \begin{align*}
        \sum_{i \in [k]} (z_i)^p - \sum_{i \in [k]} (\tilde{y}_i)^p &=\epsilon/2 \left ( p \, (c_{i''})^{p-1} -  p \, (c_{i'})^{p-1} \right) \tag{Mean Value Theorem}\\
        &\ge 0, \tag{$c_{i'} > c_{i''}, p \in (0,1]$}
    \end{align*}
    where $c_{i'} > c_{i''}$ since, (i) $c_{i'} \in (\tilde{y}_{i'}-\epsilon/2, \tilde{y}_{i'})$, by the Mean Value Theorem, (ii) $c_{i''} \in (\tilde{y}_{i''}, \tilde{y}_{i''} +\epsilon/2)$, by the Mean Value Theorem, and (iii) $\tilde{y}_{i'} - \tilde{y}_{i''} > \epsilon)$, by definition.
    
    For $p>0$ and $ \sum_{i \in [k]} (z_i)^p \ge \sum_{i \in [k]} (\tilde{y}_i)^p$, we also have that $\left( \sum_{i \in [k]} (z_i)^p\right)^{1/p} \ge \left( \sum_{i \in [k]} (\tilde{y}_i)^p\right)^{1/p}$. 
    For $p < 0$, in the above argument we would instead have $\sum_{i \in [k]} (z_i)^p \leq \sum_{i \in [k]} (\tilde{y}_i)^p$, which also implies that $\left( \sum_{i \in [k]} (z_i)^p\right)^{1/p} \ge \left( \sum_{i \in [k]} (\tilde{y}_i)^p\right)^{1/p}$. This concludes the proof of~\Cref{lemma: technical bound}.
\end{proof}

\begin{proof}[Proof of \Cref{theorem: maximin lower bound}]

 Choose any  $\mu_0 \, r \ge 2$ and a number of items $m$. Let $k = \frac{m}{\mu_0 \cdot r}$. Also let the number of total individuals $|N|$ be such that $|N| \ge \frac{m-k}{\sqrt{mk}}+1$. We will only have two centers. Center 1 has only one individual ($|I_1| = 1$) which values only items $1, \cdots , k$ equally ($v^T_{1,1} = \left[\frac{1}{\sqrt{k}}, \frac{1}{\sqrt{k}}, \cdots, 0 \right]$). Center 2 has $|N|-1$ individuals. Each individual values all items equally ($\forall i \in I_2, v^T_{2,i} = \left[\frac{1}{\sqrt{m}}, \cdots,  \frac{1}{\sqrt{m}}\right]$). Notice that due to \Cref{lemma:maxmin example incoherence}, indeed has the desired \textit{``incoherence times rank''} value.

To prove the lower bound we will first prove a lower bound on $GMW_p(y^*)$ and then provide one natural equilibrium of the derived game. Consider the following allocation $\hat{y}$ where we give items $1, \cdots, k$ to center $1$ and the rest of the items to center $2$. Then, we would have that $u_{1,1}(\hat{y}) = \sqrt{k}$ and $\forall i \in I_2, u_{2,i}(\hat{y}) = \frac{m-k}{(|N|-1)\sqrt{m}} \ge \frac{m-k}{|N|\sqrt{m}}$. However, recall that we chose $|N| \ge \frac{m-k}{\sqrt{km}}+1$ and thus $u_{1,1}(\hat{x}) \ge u_{2,i}(\hat{x})$. Thus,
\[GMW_p(y^*) \ge GMW_p(\hat{y}) = u_{2,i}(\hat{y}) \ge \frac{m-k}{|N|\sqrt{m}}.\]
On the other hand we can see that all individuals equally sharing items $1, \cdots, k$ and center $2$ getting the rest items is an equilibrium. In this equilibrium $GMW_p(\tilde{y}) = u_{1,1}(\tilde{y}) = \frac{\sqrt{k}}{|N|}$. Using the above we have that:
\[\frac{GMW_p(y^*)}{GMW_p(\tilde{y})} \ge \frac{m-k}{\sqrt{mk}} = \sqrt{\frac{m}{k}} - \sqrt{\frac{k}{m}} = \sqrt{\mu_0 r} - \frac{1}{\sqrt{\mu_0 r}} = \Omega\left(\sqrt{\ell}\right).\]

\end{proof}

\begin{proof}[Proof of \Cref{theorem:nsw poa}]
When $p = 0$, $GMW_0(y) = NSW(y) = \left(\prod_{c \in [C]} \prod_{i \in I_c} u_{c,i}(y) \right)^{1/|N|}$. Let $\tilde{y}$ be the allocation at equilibrium and $\tilde{b}$ be the corresponding bids. Let $\tilde{p}_j = \sum_{c \in [C]} \sum_{i \in I_c} \tilde{b}_{j,(c,i)}$, and $(\tilde{p}_{-c})_j = \tilde{p}_j - \sum_{i \in I_c} \tilde{b}_{j,(c,i)}$.

 We will show that for each center $c \in [C]$, there exists a deviation $b'_c$ such that the resulting (given that the rest of the centers bid their equilibrium bid) allocation satisfies $y_{j,(c,i)}' = \frac{y^*_{j,(c,i)}}{\beta_{c}}$ for all $i \in I_c$, $j \in [m]$, and for some $\beta_{c}$, where recall that $y^* \in \argmax_{y \in \Fcal} GMW_{0}(y)$.
Similar to the proof of \Cref{theorem: trading post upper bound}, consider the following function:
\[f_c(x) = \sum_{j \in [m]} \frac{(\tilde{p}_{-c})_j \sum_{i \in I_c} y^*_{j,(c,i)}}{x - \sum_{i \in I_c} y^*_{j,(c,i)}}.\]
In the interval $\left(\max_{j \in [m]}\left\{\sum_{i \in I_c} y^*_{j,(c,i)} \right\},\infty\right)$, $f_c(x)$ is continuous, decreasing and 
\[\lim_{x \rightarrow \max_{j \in [m]}\left\{\sum_{i \in I_c} y^*_{j,(c,i)} \right\}^+} f_c(x) = \infty \text{ , } \lim_{x \rightarrow \infty } f_c(x) = 0.\]
Thus, there exists a point $\beta_c  > \max_{j \in [m]}\left\{\sum_{i \in I_c} y^*_{j,(c,i)} \right\}$ such that $f_c(\beta_c) = B_c$. 

Consider the following deviation for center $c$. The bid of individual $i \in I_c$ for item $j \in [m]$ is $b'_{j,(c,i)} = \frac{y^*_{j,(c,i)} (\tilde{p}_{-c})_j}{\beta_c - \sum_{i' \in I_c} y^*_{j,(c,i')}}$. We prove that this is a valid bidding strategy, i.e., (i) all bids are non-negative, and (ii) the sum of bids is less or equal to the budget $B_c$. Condition (i) is implied by the fact that $\beta_c > \max_{j \in [m]}\left\{\sum_{i \in I_c} y^*_{j,(c,i)} \right\}$. Regarding condition (ii) we have that:

\begin{align*}
    \sum_{j \in [m]} \sum_{i \in I_c} b'_{j,(c,i)} &=  \sum_{j \in [m]} \sum_{i \in I_c} \frac{y^*_{j,(c,i)} (\tilde{p}_{-c})_j}{\beta_c - \sum_{i' \in I_c} y^*_{j,(c,i')}} = \sum_{j \in [m]}  \frac{ (\tilde{p}_{-c})_j \sum_{i \in I_c} y^*_{j,(c,i)}}{\beta_c - \sum_{i' \in I_c} y^*_{j,(c,i')}} = f_c(\beta_c) = B_c.
\end{align*}

Let $y'$ be the allocation of the trading post mechanism when center $c$ bids according to $b'_{j,(c,i)}$ and all other centers bid according to $\tilde{b}$. We will prove that $y'_{j,(c,i)} = \frac{y^*_{j,(c,i)}}{\beta_c}$, for all $i \in I_c$, $j \in [m]$. 
\begin{align*}
    y'_{j,(c,i)} &= \frac{b'_{j,(c,i)}}{(\tilde{p}_{-c})_j + \sum_{i' \in I_c} b'_{j,(c,i')}}\\
    &= \frac{\frac{y_{j,(c,i)} (\tilde{p}_{-c})_j}{\beta_c - \sum_{i' \in I_c} y_{j,(c,i')}}}{ (\tilde{p}_{-c})_j + \sum_{i' \in I_c}  \frac{y_{j,(c,i')} (\tilde{p}_{-c})_j}{\beta_c - \sum_{i'' \in I_c} y_{j,(c,i'')}}}\\
    &=\frac{\frac{y_{j,(c,i)} }{\beta_c - \sum_{i' \in I_c} y_{j,(c,i')}}}{ 1 +   \frac{ \sum_{i' \in I_c} y_{j,(c,i')}}{\beta_c - \sum_{i' \in I_c} y_{j,(c,i')}}}\\
    &= \frac{y_{j,(c,i)}}{\beta_c - \sum_{i' \in I_c} y_{j,(c,i')} + \sum_{i' \in I_c} y_{j,(c,i')}} \\
    &= \frac{y_{j,(c,i)}}{\beta_c}
\end{align*}

Since the utilities of the individuals are additive we have that $u_{c,i}(y') = \frac{1}{\beta_c} u_{c,i}(y^*)$  and the utility of the center is,
\[u_c(y') = \left(\prod_{i \in I_c} u_{c,i}(y')\right)^{1/|I_c|} = \left(\prod_{i \in I_c}\frac{1}{\beta_c} u_{c,i}(y^*)\right)^{1/|I_c|} = \left(\frac{1}{\beta_{i}^{|I_c|}}\prod_{i \in I_c} u_{c,i}(y^*)\right)^{1/|I_c|} = \frac{1}{\beta_c} u_{i}(y^*)\]
Since $\tilde{y}$ is an equilibrium, bidding strategy $y'$ should not yield more utility and thus $u_c(\tilde{y}) \ge u_c(y') = \frac{1}{\beta_c} u_{i}(y^*)$. Now from the definition of our mechanism and what we proved earlier we have that
\begin{equation}
\label{equation:2}
 \frac{1}{\beta_c}\sum_{j \in [m]} \sum_{i \in I_c} y_{j,(c,i)}^*\left(\tilde{p}_{-c})_j + \sum_{i' \in I_c} b_{j,(c,i')}'\right) = \sum_{j \in [m]} \sum_{i \in I_c} y_{j,(c,i)}'\left(\tilde{p}_{-c})_j + \sum_{i' \in I_c} b_{j,(c,i')}'\right) = \sum_{j \in [m]} \sum_{i \in I_c} b_{j,(c,i)}' = B_c   
\end{equation}
On the other hand we have that:
\begin{align*}
    \sum_{c \in [C]} \sum_{j \in [m]} \sum_{i \in I_c} y_{j,(c,i)}^*\left(\tilde{p}_{-c})_j + \sum_{i' \in I_c} b_{j,(c,i')}'\right) &=   \sum_{c \in [C]} \sum_{j \in [m]} \sum_{i \in I_c}y_{j,(c,i)}^*(\tilde{p}_{-c})_j + \sum_{c \in [C]} \sum_{j \in [m]} \sum_{i \in I_c}y_{j,(c,i)}^* \sum_{i' \in I_c} b_{j,(c,i')}'\\
    &\le \sum_{j \in [m]} \sum_{c \in [C]} \sum_{i \in I_c} y_{j,(c,i)}^* \tilde{p}_j +  \sum_{j \in [m]} \sum_{c \in [C]} \sum_{i' \in I_c} b_{j,(c,i')}'  \sum_{i \in I_c} y_{j,(c,i)}^* \tag{$(\tilde{p}_{-c})_j \le \tilde{p}_j$}\\
    &\le \sum_{j \in [m]} \sum_{c \in [C]} \sum_{i \in I_c} y_{j,(c,i)}^* \tilde{p}_j +  \sum_{j \in [m]} \sum_{c \in [C]} \sum_{i' \in I_c} b_{j,(c,i')}' \tag{$\sum_{i \in I_c} y_{j,(c,i)}^* \le 1$}\\
    &= \sum_{j \in [m]} \tilde{p}_j \sum_{c \in [C]} \sum_{i \in I_c} y_{j,(c,i)}^* + \sum_{c \in [C]} B_c \\
    &\le \sum_{j \in [m]} \tilde{p}_j + B \tag{$\sum_{c \in [C]} \sum_{i \in I_c} y_{j,(c,i)}^* \le 1$}\\
    &= \sum_{j \in [m]} \sum_{c \in [C]} \sum_{i \in I_c} \tilde{b}_{c,i}^{j} + B \tag{$\tilde{p}_j = \sum_{c \in [C]} \sum_{i \in I_c} \tilde{b}_{j,(c,i)}$} \\
    &\le 2B \tag{$\sum_{j \in [m]} \sum_{i \in I_c} \tilde{b}_{j,(c,i)} \le B_c$}
\end{align*}
Combining \cref{equation:2} with the above we have that:
\[\sum_{c \in [C]} \beta_c B_c \le 2B\]
Thus we can conclude that:
\[\frac{\left(\prod_{c \in [C]} u^{B_c}_i(y^*)\right)^{1/B}}{\left(\prod_{c \in [C]} u^{B_c}_i(\tilde{y})\right)^{1/B}} = \prod_{c \in [C]} \left( \frac{u_c(y^*)}{u_c(\tilde{y})}\right)^{B_c/B} \le \prod_{c \in [C]} \left( \beta_c \right)^{B_c/B} \le \frac{1}{B} \sum_{c \in [C]} \beta_c B_c \le 2\]
Finally notice that if we set $B_c = |I_c|$ then $B = |N|$ and for any allocation $y$ we have that:
\[\left(\prod_{c \in [C]} u^{B_c}_c(y)\right)^{1/B} = \left(\prod_{c \in [C]} \left(\prod_{i \in I_c} u_{c,i}(y) \right)^{\frac{B_c}{|I_c|}}\right)^{1/B} = \left(\prod_{c \in [C]} \prod_{i \in I_c} u_{c,i}(y) \right)^{1/|N|}\]
By combining the last two inequalities we have that:
\[\frac{NSW(y^*)}{NSW(\tilde{y})} \le 2.\]
\end{proof}

\end{document}